\documentclass[acmsmall,screen,nonacm]{acmart}
\usepackage{amsthm}
\newtheorem*{problem}{Problem}
\AtBeginDocument{}

\usepackage{microtype}
\usepackage{mathtools}
\usepackage{subcaption}
\usepackage{booktabs}
\usepackage{xspace}
\usepackage[basic]{complexity}
\usepackage{siunitx}
\usepackage{multirow}
\usepackage{graphicx}

\AtBeginDocument{%
    }

\setcopyright{acmlicensed}
\copyrightyear{2026}
\acmYear{2026}

\graphicspath{{pics/}}

\DeclareMathOperator{\In}{In}
\DeclareMathOperator{\Out}{Out}
\DeclareMathOperator{\prd}{prod}
\DeclareMathOperator{\cns}{cons}
\newcommand{\tcost}{C^{\text{tr}}}
\newcommand{\opcost}{C^{\text{op}}}
\DeclareMathOperator{\tw}{tw}
\DeclareMathOperator{\St}{state}
\DeclareMathOperator{\OPT}{OPT}
\newcommand{\layout}[1]{[#1]}

\definecolor{defblue}{rgb}{0.121,0.47,0.705}
\definecolor{defred}{rgb}{0.7,0.37,0.37}
\definecolor{defgreen}{rgb}{0.01,0.65,0.20}

\newcommand{\df}[1]{\textcolor{defblue}{\emph{#1}}}
\newcommand{\prob}[1]{\textsc{#1}\xspace}
\newcommand{\alg}[1]{\textsc{#1}\xspace}

\begin{document}

\author{Clemens Eisenhofer}
\email{clemens.eisenhofer@tuwien.ac.at}
\orcid{0000-0003-0339-1580}
\affiliation{%
    \institution{TU Wien}
    \city{Vienna}
    \state{Vienna}
    \country{Austria}
}

\author{Yuwen Jia}
\email{michelle2010tx@gmail.com}
\orcid{0009-0007-2518-2030}
\affiliation{%
    \institution{Amazon}
    \city{Cupertino}
    \state{California}
    \country{USA}
}

\author{Daniel Kroening}
\email{daniel.kroening@gmail.com}
\orcid{0000-0002-6681-5283}
\affiliation{%
    \institution{Amazon}
    \city{Cupertino}
    \state{California}
    \country{USA}
}

\author{Sergey Pupyrev}
\email{spupyrev@gmail.com}
\orcid{0000-0003-4089-673X}
\affiliation{%
    \institution{Amazon}
    \city{Cupertino}
    \state{California}
    \country{USA}
}

\renewcommand{\shortauthors}{Eisenhofer et al.}



\title[Tensor Seeks Layout]{Tensor Seeks Layout: Formalizing Layout Selection for ML~Compilers}

\begin{abstract}
    Modern machine learning compilers select memory layouts for tensors to minimize execution cost
    while satisfying hardware constraints. Layout selection is inherently global: an operator may be
    fastest under one layout while its consumers prefer another, and aligning these preferences
    requires explicit layout conversions that can regress model performance. Despite its practical
    importance, layout selection lacks a formal theoretical basis, causing current compilers
    to rely on ad-hoc heuristics.

    This paper presents the first formal study of layout selection in machine learning compilers. We
    formulate the problem as combinatorial optimization over dataflow graphs, minimizing the sum of
    operator execution costs and the per-tensor cost of these conversions. Our theoretical analysis
    shows
    that optimal layout selection is computationally hard, even when restricted to programs
    containing only matrix multiplications over two-dimensional tensors. On the positive side, we
    design an optimal polynomial-time algorithm for dataflow graphs of bounded treewidth. For
    general instances, we give a weighted MaxSAT encoding that an off-the-shelf solver can optimize.

    The formulation unifies several existing layout optimization strategies, including XLA's layout
    assignment, partition dimension selection in systolic array compilers, and layout planning in
    mobile GPU optimizers. We implement the formalization in a production compiler for an AI
    accelerator and measure the execution time of the compiled models under greedy heuristics, the
    compiler's rule-based strategy, and an optimal solver. Simple heuristics degrade execution time
    by up to $5\times$ on some workloads. Where the compiler's cost model is accurate, the solver
    matches or beats the rule-based strategy. On workloads with complex data movement it falls
    behind, and since the solver minimizes the stated objective exactly, that gap isolates cost-model error from search quality, pointing to where compiler effort actually pays off.
\end{abstract}

\maketitle

\newpage

\section{Introduction}

Executing large-scale deep learning models~\cite{BengioLH21} efficiently requires specialized
hardware accelerators such as GPUs, TPUs, and custom ASICs. Because workloads are large and models
vary, compilers have become the natural place to make performance decisions, rather than hand-tuned
libraries. In practice, a large fraction of end-to-end speedups in training and
inference come not from new arithmetic kernels, but from reducing and reshaping data
movement~\cite{datamove,xu2025,attention,PagedAttention}.

Modern machine learning (ML) compilers~\cite{tvm18,mlir,triton,taso} perform various optimizations,
including operator fusion, tiling, and instruction
scheduling~\cite{autotvm,ansor,roller,rammer,flextensor}. Among these, layout
optimization, that is, the selection of how multi-dimensional tensors are arranged in memory or
mapped to hardware resources, has a direct impact on execution
performance~\cite{smartmem,neocpu,welder,alt}. A \emph{tensor layout} specifies the physical
representation of a multi-dimensional array on the target
hardware~\cite{linearlayouts,layoutabstractions}. For instance, a
four-dimensional tensor storing batched images can be laid out as \texttt{NHWC} (batch, height,
width, channels) or \texttt{NCHW} (batch, channels, height, width), with each ordering favoring
different access patterns~\cite{LXY22,vtensor}. Operators are
layout-sensitive~\cite{tensorcomp,layoutabstractions}: a convolution may prefer \texttt{NHWC} on one
backend and \texttt{NCHW} on another~\cite{cudnn,LXY22}; a matrix multiplication may admit two
symmetric variants corresponding to transposing its operands; an accelerator instruction may require
a specific dimension to be partitioned across on-chip processing elements or distributed across
multiple devices~\cite{autosa,spatial,alpa}. When a producer emits a tensor in a layout incompatible
with its consumer's expectation, the compiler must insert an explicit layout conversion, often
implemented as DMA (Direct Memory Access) transfers and reformatting code.
Whether a conversion must be materialized at all is itself target-dependent. On CPUs and GPUs,
a kernel can sometimes absorb a mismatch by iterating over the data in a different order: BLAS
routines accept transposed operands, and loop interchange can emulate a layout change. On systolic
array accelerators, by contrast, the dimension streamed into the compute array is fixed by the
hardware datapath, so conversions involve physical data movement that cannot be folded into the
consuming operator. We discuss this distinction in Section~\ref{sec:conversions}. As we argue in
Section~\ref{sec:evaluation}, these conversions add measurable overhead, which can lead to up to
$5\times$ slowdown in model execution. The extra
operators can also inflate compilation time by increasing graph size and reducing opportunities for
operator fusion~\cite{datamove,NiuGWAR21}.

\subsection{Layout Selection}

\begin{figure}[!t]
    \centering
    \begin{subfigure}[b]{.22\linewidth}
        \centering
        \includegraphics[page=5,height=5cm]{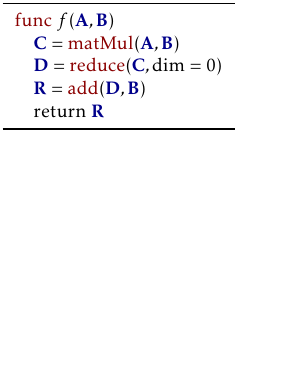}
        \caption{ML model}
        \label{fig:ex1}
    \end{subfigure}
    \begin{subfigure}[b]{.38\linewidth}
        \centering
        \includegraphics[page=6,height=5cm]{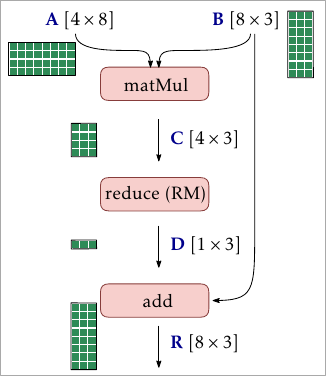}
        \caption{}
        \label{fig:ex2}
    \end{subfigure}
    \begin{subfigure}[b]{.38\linewidth}
        \centering
        \includegraphics[page=7,height=5cm]{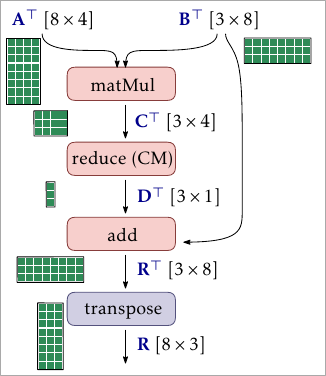}
        \caption{}
        \label{fig:ex3}
    \end{subfigure}
    \caption{Two layout assignments for the computation in~(a). Layouts are row-major (RM) or
        column-major~(CM); dimensions are in storage order. (b)~All tensors remain in RM.
        (c)~Inputs are transposed to CM, and the output tensor $R$ is transposed back.}
    \Description[Three panels: a small tensor program and two layout assignments for
    it]{Panel~(a) lists a program \texttt{func f(A,B)} computing
        $C=\mathrm{matMul}(A,B)$, then $D=\mathrm{reduce}(C,\mathrm{dim}=0)$, then
        $R=\mathrm{add}(D,B)$, returning~$R$. Panels~(b) and~(c) show the same dataflow graph
        twice, with operators as boxes, tensors as labelled arrows, and each tensor's element
        grid drawn beside it in storage order. In~(b) every tensor is row-major: inputs
        $A\,[4\times 8]$ and $B\,[8\times 3]$ feed \texttt{matMul}, producing $C\,[4\times 3]$;
        \texttt{reduce} in its RM form produces $D\,[1\times 3]$; \texttt{add} combines $D$ with
        $B$ (broadcasting $D$ over the rows of $B$) to give $R\,[8\times 3]$, and no conversion
        is needed. In~(c) the inputs arrive transposed as $A^\top[8\times 4]$ and
        $B^\top[3\times 8]$, so \texttt{matMul} yields $C^\top[3\times 4]$, the cheaper CM form
        of \texttt{reduce} yields $D^\top[3\times 1]$, and \texttt{add} yields
        $R^\top[3\times 8]$; a final \texttt{transpose} operator, shaded differently from the
        compute operators, converts $R^\top$ back to $R\,[8\times 3]$.}
    \label{fig:layout-example}
\end{figure}

This paper studies tensor layout selection as a first-class compilation problem. To illustrate,
consider the computation shown in Figure~\ref{fig:ex1}: a matrix multiplication followed by a
reduction and an element-wise addition, where input tensor $B$ is shared between the \texttt{matmul}
and the \texttt{add}. Each tensor can be stored in row-major (RM) or column-major (CM) order, and
each operator's execution time depends on the layouts of its inputs and outputs. On many
accelerators, the \texttt{matmul} admits two equivalent forms, $C = A \times B$ and
$C^\top = B^\top \times A^\top$, that produce $C$ in different layouts. The \texttt{reduce} operator
collapses the matrix $C$ to a vector by summing along columns, and its cost depends on the layout:
columns are strided in RM but contiguous in CM. The \texttt{add} requires its
operands to share a layout, coupling $D$'s layout back to~$B$'s.

Figure~\ref{fig:layout-example} compares two layout assignments for this computation. In~(b), all
tensors use RM: the \texttt{matmul} is cheap but the \texttt{reduce} pays the strided-column cost.
In~(c), the function takes transposed inputs and applies the transposed \texttt{matmul}, so the reduction
becomes cheap at the cost of an explicit layout conversion for the output tensor,~$R$.
Which assignment is preferable depends on the costs: if the \texttt{transpose} costs more than the
savings on the reduction, (b) is the better choice; otherwise, (c) is. In real
models with hundreds or thousands of operators, such dependencies multiply across the dataflow
graph. It is the compiler's task to select layouts for all tensors such that the total operator and
data transformation cost is minimized.

This challenge arises across diverse hardware targets and compilation contexts:
\begin{itemize}
    \item {\bf Systolic array accelerators} (e.g., TPUs, AWS Trainium) distribute tensors
    across parallel processing elements by splitting along a chosen dimension, called the
    \emph{partition dimension}. Different choices of partition dimension enable different
    instruction sequences with varying efficiency~\cite{FuZFHEHSWD024,neuron_compiler}.

    \item {\bf GPU backends} in compilers like XLA and TVM must choose between data formats such as \texttt{NCHW}
    and \texttt{NHWC} for convolutions, with the preferred choice depending on convolution
    parameters and downstream operations~\cite{xla,tvm18}.

    \item {\bf Mobile and edge deployments} operate under tight memory budgets where unnecessary
    layout conversions can degrade performance or prevent model execution
    altogether~\cite{smartmem,neocpu}.
\end{itemize}

In general, given a dataflow graph, each tensor admits a finite set of feasible layouts, each
operator has a set of admissible input/output layout combinations with associated performance cost,
and converting a tensor between layouts has an explicit cost. The compiler must assign layouts
across the graph to minimize the total execution time of operators plus the time needed for inserted
conversions, subject to hardware feasibility constraints. Although many tools implicitly target this
objective, it is rarely formalized~\cite{xla,tvm18,onednn,smartmem,neocpu}. Instead, most modern
compilers rely on a chain of heuristic graph passes: pick a ``good'' layout locally, propagate
preferences forward and backward, and resolve remaining conflicts by inserting transposes.
Unfortunately, such local approaches can produce poor results when operator preferences conflict
across the graph, as an example in Figure~\ref{fig:local_vs_global} illustrates. Each \texttt{Conv}
prefers \texttt{NHWC} and each \texttt{BatchNorm} prefers \texttt{NCHW}, so a per-operator strategy
assigns different layouts to adjacent operators. Aligning them for the downstream \texttt{Add} then
requires a \texttt{Transpose} after every \texttt{BatchNorm} (Figure~\ref{fig:localB}). A global
view gives a better solution: choosing the less efficient \texttt{NHWC} for each \texttt{BatchNorm}
removes all transposes (Figure~\ref{fig:localC}).

\subsection{Our Contributions}

\definecolor{colconv}{rgb}{0.97,0.81,0.8}
\definecolor{colbn}{rgb}{0.99,0.9,0.8}
\definecolor{coladd}{rgb}{0.84,0.91,0.83}

\begin{figure}[!t]
    \centering
    \begin{subfigure}[b]{.15\linewidth}
        \centering
        \includegraphics[page=4,width=0.95\textwidth]{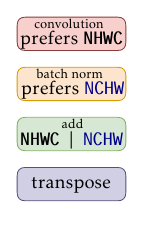}
    \end{subfigure}
    \begin{subfigure}[b]{.4\linewidth}
        \centering
        \includegraphics[page=2,width=0.95\textwidth]{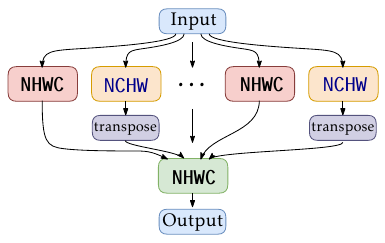}
        \caption{Locally-optimal layout}
        \label{fig:localB}
    \end{subfigure}
    \begin{subfigure}[b]{.4\linewidth}
        \centering
        \includegraphics[page=3,width=0.95\textwidth]{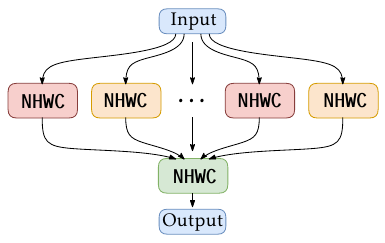}
        \caption{An optimized layout}
        \label{fig:localC}
    \end{subfigure}
    \caption{Local versus global layout selection on a graph of
        \colorbox{colconv}{\texttt{Conv}} and \colorbox{colbn}{\texttt{BatchNorm}} operators
        feeding a shared \colorbox{coladd}{\texttt{Add}}. Layouts are \texttt{NHWC} or
        \texttt{NCHW}; inserted \texttt{Transpose} operators are shown explicitly.}
    \Description[Legend plus two layout assignments of a fan-out/fan-in graph, one with
    transposes and one without]{The leftmost panel is a legend of four operator boxes:
        convolution (prefers \texttt{NHWC}), batch norm (prefers \texttt{NCHW}), add (accepts
        \texttt{NHWC} or \texttt{NCHW}), and transpose. The two remaining panels each show the
        same graph: a single \texttt{Input} node fanning out to a row of alternating convolution
        and batch-norm operators (an ellipsis marks the omitted middle ones), which all fan back
        in to one shared \texttt{Add} that feeds \texttt{Output}. In the locally-optimal
        assignment~(b), each operator takes its own preferred layout, so the convolutions are
        \texttt{NHWC} and the batch-norms \texttt{NCHW}; because the shared \texttt{Add} runs in
        \texttt{NHWC}, a \texttt{transpose} box is inserted on the outgoing edge of every
        batch-norm. In the globally optimized assignment~(c), every operator including the
        batch-norms runs in \texttt{NHWC}, so all edges reach the \texttt{Add} directly and no
        transpose boxes appear at all.}
    \label{fig:local_vs_global}
\end{figure}

Layout optimization requires an explicit mathematical model rather than an ad-hoc
collection of heuristics. We propose a model built from a finite
set of layout choices for each tensor, a cost function expressing the execution performance for
each layout combination, and transpose costs for converting tensors between layouts. It
subsumes several existing ``ad hoc'' layout optimizations under a single objective. XLA's
layout assignment chooses minor-to-major dimension orderings while paying copy costs to reconcile
mismatches~\cite{xla,xla_layout_impl}. On systolic array accelerators, compilers such as AWS Neuron
Graph Compiler~(see Appendix~\ref{app:compiler} and \cite{neuron_compiler,FuZFHEHSWD024}) select
partition dimensions for tensors in on-chip buffers, which maps directly to our formulation.
Heuristic systems such as SmartMem~\cite{smartmem} and NeoCPU~\cite{neocpu} correspond to restricted
variants where layout domains are pruned and layout-flexible operators are treated as nearly
indifferent, so propagation approximately minimizes transpose costs. Stating all four under one
objective lets us compare their design choices and the restriction each one makes for
tractability.

\prob{Layout Selection} is \NP-hard even when each operator has only a constant number of inputs and
each tensor admits only a small number of layouts, and hardness persists in restricted
settings that resemble common matmul-heavy graphs. This rules out a universal polynomial-time
algorithm with strong worst-case guarantees. The formulation does, however, expose structure
that can be exploited. For dataflow graphs with bounded \emph{treewidth} (a structural property of
graphs measuring their ``tree-likeness''~\cite{RS1984,Bodlaender93}), the problem admits an optimal
polynomial-time algorithm; that is, the problem is \emph{fixed-parameter tractable} with respect to
treewidth together with the number of layout configurations per operator. Such tractability is a
known phenomenon in algorithm
design~\cite{Courcelle90,ArnborgLS91,Thorup98}; our contribution is an explicit dynamic program for
\prob{Layout Selection} that is directly implementable inside a compiler
(Section~\ref{sec:treewidth}). For general instances, we encode the problem
as a weighted \prob{MaxSAT} instance and invoke a general-purpose
solver~\cite{Z3,DBLP:series/faia/LiM21}.

We validate the formulation by integrating it into a production graph compiler for an AI
accelerator, AWS Trainium~\cite{aws_trainium,neuron_compiler}. The prototype instantiates layout
domains and cost functions using a cycle-based cost model and hardware feasibility checks. We then
evaluate a range of algorithms: local and greedy baselines that mirror common compiler
practice, a manual rule-based strategy, and exact solver-based encodings that compute optimal
solutions, enabling direct measurement of heuristic optimality gaps. Across a diverse suite of
models, we find that (i)~layout selection has a substantial impact on performance, with simple
heuristics degrading execution time by up to $5\times$ compared to a domain-aware strategy,
(ii)~solver-based methods are viable for the majority of models, completing within seconds on most
instances, and (iii)~the benefit of exact solving is limited by cost model fidelity.
Sections~\ref{sect:analysis} and~\ref{sect:dev} take up the third point.

In summary, our main contributions are as follows:

\begin{enumerate}
    \item A unified formulation of automatic layout selection for tensors as a graph
    optimization problem parameterized by layout domains, operator costs, and transpose costs.

    \item Algorithmic results establishing \NP-hardness and inapproximability in general, an
    explicit fixed-parameter algorithm for graphs of bounded treewidth, and a weighted \prob{MaxSAT}
    encoding for general instances.

    \item A framework and evaluation in a production compiler setting,
    comparing heuristic and solver-based approaches for tensor layout selection, together with an
    analysis methodology that uses provably optimal solutions to separate algorithmic quality from
    cost-model fidelity.
\end{enumerate}

\section{A Mathematical Formulation of Layout Selection}
\label{sec:formulation}

This section develops a model of layout selection that formalizes its constraints and
costs.

\paragraph{Tensors, operators, and dataflow graphs.}
A \df{tensor} is a multi-dimensional array, and an \df{operator} is a function on tensors that takes
as input a tuple of tensors, denoted $\In(o)$, and produces output tensors, denoted $\Out(o)$. We
model an ML program as a directed acyclic \df{dataflow graph} $G=(V,E)$, a standard representation in
optimizing compilers~\cite{halide,mirage,tvm18}. Let $\mathcal{T}$ denote the set of all tensors in
$G$. Each vertex $o \in V$ represents an operator, and each directed edge $(o_1,o_2)\in E$ carries a
tensor $t \in \mathcal{T}$ produced by $o_1$ and consumed by $o_2$. We treat $E$ as a multiset and
distinguish repeated uses of the same tensor by argument position. We assume each operator has at
most $\Delta$ inputs, that is, $|\In(o)| \leq \Delta$ for every operator $o$; in practice $\Delta$
is a small constant. We also assume exactly one output, that is, $|\Out(o)| = 1$, which is without
loss of generality: operators with multiple outputs can be handled by introducing auxiliary
single-output operators, each producing one of the outputs.

\paragraph{Tensor layouts.}
A \df{layout} of a tensor specifies how its elements are represented on the target hardware, such as
a permutation of dimensions (e.g., row-major vs.\ column-major), memory placement (e.g., shared vs.\
global memory), or a partition dimension for on-chip buffers. We abstract these details by assuming
that each tensor $t$ has a finite set of admissible layouts $L(t)$. The size of $L(t)$ is determined
by the layout interpretation on the target. For a $k$-dimensional tensor, there are, for example, at
most $k$ choices of a partition dimension and at most $k!$ dimension orderings; in practice the set
is small. For a sequence of tensors $T=\langle t_1,\dots,t_r\rangle$, we write
$L(T)=L(t_1)\times\cdots\times L(t_r)$ for the set of joint layout choices for $T$. A layout
assignment (defined formally below) selects, for every operator $o$ and every tensor $t$ that is an
input or output of $o$ (i.e., $t \in \In(o)\cup\Out(o)$), a layout $\ell(t,o) \in L(t)$ in which $o$
accesses $t$. For a sequence $T=\langle t_1,\dots,t_r\rangle$ of input or output tensors of $o$, we
write $\ell(T,o)$ for the tuple $\bigl(\ell(t_1,o),\dots,\ell(t_r,o)\bigr)$ of layouts selected at
$o$; for example, $\ell(\In(o),o)$ is the tuple of layouts in which $o$ reads its inputs.

\paragraph{Operator costs and feasibility.}
Operators impose constraints and preferences on the layouts of their inputs and outputs. We model
this via a set of non-negative \df{operator cost} functions
\[
\opcost_o :
L(\In(o)) \times L(\Out(o)) \to \mathbb{R}_{\ge 0} \cup \{\infty\},
\]
where $\In(o)$ and $\Out(o)$ denote the input and output tensors of $o$. A cost of $\infty$ denotes
an infeasible or unsupported layout combination on the target hardware. Finite values encode
relative performance differences between supported layout choices.

\paragraph{Layout conversion costs.}
When a tensor is produced in one layout but consumed in one or more different layouts, the compiler
must insert an explicit \df{layout conversion} (a transpose, copy, or reformat). For each tensor
$t$, we model its cost using a \df{transpose cost} function
\[
\tcost_t : L(t) \times L(t) \to \mathbb{R}_{\ge 0},
\]
with $\tcost_t(\ell,\ell)=0$. This cost abstracts the data movement, reformatting overhead, and
associated synchronization required to convert $t$ from one layout to another. When multiple
consumers require the same layout, a single transpose suffices and its cost is incurred once. In
practice, transpose costs can be estimated via a cost model.

\paragraph{The layout selection problem.}
A \df{layout assignment} chooses a layout $\ell(t,o) \in L(t)$ for every operator $o$ and each
tensor $t \in \In(o)\cup\Out(o)$. For each tensor $t$, let $\prd(t)$ denote the layout
$\ell(t,o_{\text{prod}})$ assigned at its producer $o_{\text{prod}}$ and let \( \cns(t) =
\{\,\ell(t,o) \mid o \text{ consumes } t\,\} \) denote the set of distinct layouts required by its
consumers. This formulation captures the tradeoff between choosing layouts that are efficient for
individual operators and avoiding expensive layout conversions.

\begin{problem}[\prob{Layout Selection}]
    Given a dataflow graph $G=(V,E)$, feasible layout sets $L(t)$ for each
    tensor~$t \in \mathcal{T}$, operator and transpose costs, find layouts $\ell(t,o) \in L(t)$ for
    every operator $o$ and tensor $t \in \In(o)\cup\Out(o)$ minimizing
    \[
    \sum_{o\in V}
    \opcost_{o}\!\bigl(\ell(\In(o),o),\,\ell(\Out(o),o)\bigr)
    \;+\;
    \sum_{t\in \mathcal{T}}
    \sum_{\ell \in \cns(t)}
    \tcost_{t}\!\bigl(\prd(t),\,\ell\bigr).
    \]
\end{problem}

The first sum captures operator execution costs. The second captures layout conversion costs: for
each tensor $t$, the compiler converts from $\prd(t)$ to each distinct layout in $\cns(t)$. Since
$\tcost_t(\ell,\ell)=0$, consumers that agree with the producer incur no cost.

The formulation abstracts over the specific interpretation of tensor layouts, accommodating dimension
orderings, partition dimensions, memory placements, and combinations thereof. This generality
enables our theoretical results and algorithms to apply across diverse hardware targets. For
example, our model subsumes VTensor's formulation~\cite{vtensor} as a special case with uniform
operator costs, and provides a formal objective that systems such as SmartMem~\cite{smartmem} and
NeoCPU~\cite{neocpu} heuristically optimize. Next, we describe two concrete instantiations of
\prob{Layout Selection} in existing ML compilers.

\subsection{Instantiation: Layout Assignment in XLA}

Our formulation directly captures the layout assignment problem in XLA and similar GPU-oriented
compilers~\cite{xla,tvm18}. In this context, a tensor layout specifies the \emph{minor-to-major
    ordering} of dimensions in memory. A $k$-dimensional tensor admits up to $k!$ layouts, which are a
subset of permutations on $\langle 1, \dots, k \rangle$, forming $L(t)$. For instance, a
4-dimensional tensor in \texttt{NHWC} format has layout $\langle C, W, H, N \rangle$ while
\texttt{NCHW} corresponds to $\langle W, H, C, N \rangle$.

The operator cost function encodes hardware preferences: convolutions on GPUs typically perform best
with specific layouts (e.g., \texttt{NHWC} for certain cuDNN kernels), while other layouts incur
overhead or may be unsupported (that is, have $\infty$ cost). Matrix multiplications prefer layouts
where the inner (contracted) dimension is contiguous in memory. Layout-agnostic operators like
elementwise operations have constant cost across all layouts. The transpose cost corresponds to
inserting a \emph{copy} operation that physically rearranges tensor data in memory; XLA's layout
assignment pass inserts such copies whenever adjacent operators require incompatible
layouts~\cite{xla_layout_impl}. Our formulation thus provides a precise characterization of the
optimization problem that XLA's greedy heuristic attempts to solve.

\subsection{Instantiation: Partition Dimension Selection in AWS Trainium}

Our formulation also captures layout optimization in compilers for systolic array accelerators such
as AWS Trainium~\cite{FuZFHEHSWD024,neuron_compiler,aws_trainium}; refer to
Appendix~\ref{app:compiler} for details. In these systems, tensors residing in on-chip scratchpad
memory must specify a \emph{partition dimension}, that is, the axis along which the tensor is
distributed across parallel processing elements. For a $k$-dimensional tensor, admissible layouts
are $L(t) = \{1, 2, \ldots, k\}$, indicating which dimension serves as the partition dimension.

The operator cost function encodes hardware constraints imposed by the systolic array architecture.
Matrix multiplication on Trainium requires specific partition dimensions for its operands: computing
$C = A \times B$ requires $A$ to be partitioned along its second dimension and $B$ along its first
dimension. The cost function assigns finite values to valid partition configurations (derived from a
cost model reflecting instruction latency) and $\infty$ to invalid configurations. The transpose
cost corresponds to layout conversions across partitions (typically implemented via
dedicated hardware engines), an overhead that the compiler seeks to minimize.

\subsection{Layout Conversions versus Loop Transformations}
\label{sec:conversions}

On general-purpose hardware, explicit layout conversions are often avoidable: many of the benefits of
a different physical layout can be handled by changing the \emph{iteration order} of the consuming kernel
instead. BLAS interfaces such as \texttt{cublasSgemm} accept transposed operands via \texttt{op(A)}
and \texttt{op(B)} flags, folding the conversion into the kernel itself, and on CPUs a loop
interchange (e.g., the classic $ikj$ matmul ordering) restores unit-stride access to a row-major
operand without ever materializing its transpose. The dependency also runs the other way: physically
reordering data can be what \emph{enables} a loop optimization, since transposing or packing operands
exposes the contiguous access patterns that let compilers vectorize transformer
kernels~\cite{vectorizing-transformers}, and blocked formats serve the same purpose for CNN inference
on CPUs~\cite{neocpu}. Whether the layout or the loop nest should give way depends on the target.

Our cost framework is agnostic to this choice. If a kernel can absorb a layout mismatch by iterating
differently, that combination is simply another admissible configuration of the operator, with a
finite (possibly higher) operator cost and no transpose cost. If it cannot, the mismatch requires an
explicit conversion and the transpose cost is charged.

On the systolic array accelerators we target, a kernel generally \emph{cannot} absorb a mismatch in
the partition dimension. As detailed in Appendix~\ref{app:compiler}, each of the 128 on-chip memory partitions
feeds exactly one row of the systolic array, so the tensor dimension streamed into the array must be
the one laid out across partitions. No reordering of the surrounding loop nest can change which
dimension that is. Switching the partition dimension requires physically moving every element to a
different partition, an operation performed by dedicated transpose hardware or DMA engines and
impossible to fuse into the consuming matmul instruction (in contrast to the
\texttt{op(A)}/\texttt{op(B)} flags above). The same holds for related targets such as TPUs, where
tensors are blocked into fixed hardware tiles~\cite{tpugraphs}. This is why our formulation charges
transpose costs explicitly, while their \emph{magnitude} is target-specific (see
Section~\ref{sec:threats}).

\section{Theoretical Considerations}
\label{sec:theory}

Unless stated otherwise, we consider the general formulation from
Section~\ref{sec:formulation}, with a constant upper bound $\Delta$ on operator arity.

\subsection{Computational Hardness}
\label{sec:hardness}

\begin{theorem}
    \label{thm:np-hard}
    The \prob{Layout Selection} problem is $\NP$-hard, even when every operator has at most three inputs ($\Delta=3$).
\end{theorem}

\begin{proof}
    We reduce from \prob{3-SAT}~\cite{GJ79}. Given a formula, we build an instance of \prob{Layout Selection} in which variables correspond to tensors and clauses correspond to operators, such that the formula is satisfiable if and only if the minimum cost is~$0$.

    Every tensor has two layouts, $\mathsf{T}$ (``true'') and $\mathsf{F}$ (``false''). For every tensor $t$, set $\tcost_t(\mathsf{T},\mathsf{F}) = 1$ and $\tcost_t(\mathsf{F},\mathsf{T}) = 1$.
    For each Boolean variable $x_j$, create a source operator $s_j$ with no inputs and output tensor $t_j$, with $\opcost_{s_j} = 0$ for both layouts. For each clause $C_i$ over variables $x_{j_1}, x_{j_2}, x_{j_3}$, create an operator $o_i$ whose three inputs are tensors $t_{j_1}, t_{j_2}, t_{j_3}$ and whose output is a fresh tensor used nowhere else. Define $\opcost_{o_i}$ to be $0$ exactly on those input-layout triples that satisfy $C_i$ (where layout $\mathsf{T}$ represents a positive literal and $\mathsf{F}$ a negative literal), and $1$ otherwise; the output layout does not affect cost.

    If the formula is satisfiable, set $\ell(t_j,s_j)$ according to a satisfying assignment; then every clause-operator has zero cost and all consumers of each $t_j$ agree with its producer, so no transpose cost is incurred and the total cost is~$0$. Conversely, if there exist layouts
    $\ell(t,o)$ with total cost~$0$, then no tensor incurs a transpose cost, so all consumers of each $t_j$ use the same layout as its producer, and every operator must have operator cost~$0$,
    which implies the induced truth assignment satisfies every clause.
\end{proof}

Note that the proof above implies that it is $\NP$-hard to distinguish instances of optimum cost $0$
from instances of optimum cost at least $1$. Since any solution with positive cost has an unbounded
ratio relative to an optimum of $0$, the hardness result rules out approximation algorithms with any
finite guarantee.

\begin{corollary}
    \label{cor:no-approx}
    Unless $\P=\NP$, \prob{Layout Selection} admits no polynomial-time approximation algorithm with
    a finite approximation ratio.
\end{corollary}

Given this negative result, finding algorithms with strong guarantees requires structural or cost
restrictions. One apparent strategy is limiting the operator types. However, as the next theorem
shows, \prob{Layout Selection} remains computationally hard even when all operators are matrix
multiplications with $\Delta=2$. Note that the inapproximability result of
Corollary~\ref{cor:no-approx} does not follow from this restricted setting, since the underlying
\prob{Odd Cycle Transversal} problem admits approximation algorithms~\cite{AgarwalCMM05}; this is
why the two reductions are complementary.

\begin{figure}[!t]
    \centering
    \begin{subfigure}[b]{.42\linewidth}
        \centering
        \includegraphics[page=1]{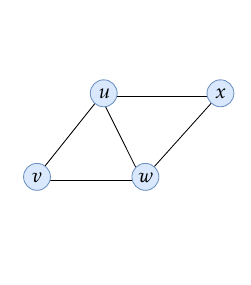}
        \caption{Graph $H=(U,F)$}
    \end{subfigure}
    \hfill
    \begin{subfigure}[b]{.55\linewidth}
        \centering
        \includegraphics[page=2,height=6cm]{algorithms}
        \caption{Constructed \prob{Layout Selection} instance}
    \end{subfigure}
    \caption{Reduction from \prob{Odd Cycle Transversal} to \prob{Layout Selection}
        (Theorem~\ref{thm:matmul-hard}). (a)~An undirected graph $H$ with a triangle
        $\{u,v,w\}$.
        (b)~The corresponding ML program: each vertex becomes a tensor with two layouts
        $\{\mathsf{A},
        \mathsf{B}\}$, and each edge becomes a matmul requiring its inputs in different layouts
        (cost
        $0$) or infeasible (cost $\infty$). The triangle forms an odd cycle of opposite-layout
        constraints that cannot be two-colored with one layout per tensor; at least one tensor must
        be read in both layouts, incurring a transpose cost of~$1$.}
    \Description[A four-vertex graph and the layout-selection program built from it]{Panel~(a)
        draws an undirected graph on vertices $u$, $v$, $w$, $x$ with five edges: $\{u,v\}$,
        $\{v,w\}$, $\{u,w\}$, $\{u,x\}$ and $\{w,x\}$. Vertices $u$, $v$, $w$ form a triangle,
        that is, an odd cycle, and $u$, $w$, $x$ form a second one. Panel~(b) lists the
        constructed program \texttt{func f()}: first one source operator per vertex, producing
        tensors $t_u$, $t_v$, $t_w$, $t_x$, each with layout drawn from $\{\mathsf{A},
        \mathsf{B}\}$; then one matmul per edge, $r_{uv}=\mathrm{matmul}(t_u,t_v)$,
        $r_{vw}=\mathrm{matmul}(t_v,t_w)$, $r_{uw}=\mathrm{matmul}(t_u,t_w)$,
        $r_{ux}=\mathrm{matmul}(t_u,t_x)$ and $r_{wx}=\mathrm{matmul}(t_w,t_x)$, each annotated
        with the edge it encodes. Two trailing comments record the costs: a matmul costs $0$ if
        and only if its two input layouts differ, and the transpose cost between $\mathsf{A}$ and
        $\mathsf{B}$ is $1$ in either direction.}
    \label{fig:reduction}
\end{figure}

\begin{theorem}
    \label{thm:matmul-hard}
    \prob{Layout Selection} is $\NP$-hard even when all non-source operators are matrix
    multiplications over two-dimensional tensors, each tensor admits exactly two layouts, transpose
    costs are binary, and operator costs are $0$ or $\infty$.
\end{theorem}

\begin{proof}
    We reduce from \prob{Odd Cycle Transversal} (deleting the fewest vertices to make a graph
    bipartite). Given an undirected graph $H=(U,F)$, create a tensor $t_u$ for each vertex
    $u\in U$ with exactly two layouts $\{\mathsf{A},\mathsf{B}\}$. For each vertex $u$, create a
    source operator $s_u$ with no inputs and output $t_u$, with $\opcost_{s_u}=0$ for both layouts.
    For each edge $\{u,v\}\in F$, create a matrix-multiplication operator $o_{uv}$ whose inputs are
    $(t_u,t_v)$ and whose output is a fresh tensor used nowhere else. Define $\opcost_{o_{uv}}$ to
    be $0$ when $t_u$ and $t_v$ use different layouts (capturing the two symmetric matmul variants),
    and $\infty$ otherwise. For every tensor $t_u$, set
    $\tcost_{t_u}(\mathsf{A},\mathsf{B})=\tcost_{t_u}(\mathsf{B},\mathsf{A})=1$.
    Figure~\ref{fig:reduction} illustrates the construction.

    Call $u$ \df{split} if its incident matmuls read $t_u$ in both layouts. Since its source layout
    is free, $t_u$ costs $1$ exactly when $u$ is split. If $W$ is the set of split vertices, the
    single layouts of vertices outside $W$ two-color $H-W$. Conversely, any bipartition of $H-W$
    satisfies edges outside $W$, while endpoints in $W$ can use the opposite layout per incident
    edge. Hence the optimum equals the minimum odd-cycle transversal size. Since \prob{Odd Cycle
        Transversal} is $\NP$-hard~\cite{GJ79}, the result follows.
\end{proof}

\subsection{Optimal Algorithms for Layout Selection}
\label{sec:opt}

Hardness in general does not preclude solving the instances that actually arise, and we give two exact
algorithms for \prob{Layout Selection}. The first exploits the
structural property of bounded \emph{treewidth} in real-world dataflow
graphs~\cite{RS1984,Bodlaender93}: when the treewidth is bounded, \prob{Layout Selection} can be
solved in polynomial time via dynamic programming (DP). Prior work suggests that real-world ML
graphs often have small treewidth~\cite{0001PSVW19}; we confirm the observation empirically in
Section~\ref{sec:evaluation}. The second approach encodes the problem as a weighted \prob{MaxSAT}
instance~\cite{DBLP:series/faia/LiM21}, enabling off-the-shelf solvers to find optimal solutions for
general instances.

\subsubsection{Treewidth-based Exact Algorithm}
\label{sec:treewidth}
\leavevmode\par

\paragraph{Treewidth.}
Let $H=(U,F)$ be an undirected graph. A \df{tree decomposition} of $H$ is a pair $(T,\{X_i\})$
consisting of a tree $T$ with node set $V(T)$ and, for each tree node $i \in V(T)$, a set
$X_i \subseteq U$ of graph vertices (called a \df{bag}), such that:
\begin{enumerate}
    \item $\bigcup_{i\in V(T)} X_i = U$;
    \item for every edge $\{u,v\}\in F$ there exists $i$ with $\{u,v\}\subseteq X_i$;
    \item for every $u\in U$, the set $\{ i \mid u\in X_i \}$ induces a connected subtree of $T$.
\end{enumerate}
The \df{width} of the decomposition is $\max_i |X_i|-1$. The \df{treewidth} $\tw(H)$ is the minimum
width over all tree decompositions of $H$. Intuitively, a tree decomposition organizes the graph
into a tree of small separators, where each bag separates the vertices appearing in the subtree below
it from the rest of the graph. A dynamic programming algorithm can therefore sweep over the tree
bottom-up and, at any point, reason only about the small ``frontier'' of vertices in the current
bag.
Such a sweep is simplest when the tree has a restricted form.
A tree decomposition is
\df{nice} if $T$ is rooted at a node~$r$ with $X_r=\emptyset$, and each tree node $i$ has one of the
following four types:
\begin{itemize}
    \item \textbf{Leaf:} no children, $X_i=\emptyset$.
    \item \textbf{Introduce:} one child $j$, $X_i = X_j \cup \{v\}$ for some vertex
    $v \in U \setminus X_j$ (vertex $v$ enters the frontier).
    \item \textbf{Forget:} one child $j$, $X_i = X_j \setminus \{v\}$ for some $v \in X_j$
    (vertex $v$ leaves the frontier, never to reappear above $i$).
    \item \textbf{Join:} two children $j_1,j_2$ with $X_i = X_{j_1}=X_{j_2}$
    (two independently processed subtrees are merged).
\end{itemize}
Any tree decomposition can be transformed into a nice one of the same
width, with $O(\tw\cdot|U|)$ nodes, in linear time~\cite{Kloks94}. Its sole purpose is to simplify the dynamic
program; since each step changes the bag by at most one vertex, the algorithm only needs transition
rules for the four node types above. In a nice tree decomposition, every vertex of the underlying
graph is introduced exactly once and forgotten exactly once. Figure~\ref{fig:treewidth} illustrates
a nice tree decomposition for the example in Figure~\ref{fig:layout-example}. While computing the
treewidth of a graph is itself \NP-hard in general, it is fixed-parameter tractable, and practical
exact and heuristic tools produce decompositions of near-optimal width for graphs of small
width~\cite{Bodlaender93}; we report the treewidth of our benchmark graphs in
Section~\ref{sec:evaluation}.

To simplify the presentation and make the problem amenable to dynamic programming over a tree
decomposition, we reformulate \prob{Layout Selection} as a constrained vertex labeling problem. For
each operator $o\in V$, define its local \df{state space}:
\[
\St(o) \;=\; L\bigl(\In(o)\bigr)\times L\bigl(\Out(o)\bigr).
\]
Since $|\In(o)|\le \Delta$ and each $L(t)$ is finite, we have $|\St(o)|\le M$ for a constant $M$
that depends only on $\Delta$ and the maximum number of admissible layouts per tensor, but is
independent of $|V|$.

A layout assignment is equivalently a labeling $g\colon V\to\biguplus_{o\in V}\St(o)$ with
$g(o)\in \St(o)$ for every $o$. We write $g(o)|_t$ for the component of $g(o)$ that gives the layout
of tensor~$t$: the output layout if $t\in\Out(o)$, or the corresponding input layout if
$t\in\In(o)$. If $o_{\text{prod}}$ is the producing operator of tensor $t$, then the producer layout
is $\prd(t)=g(o_{\text{prod}})|_t$.

\begin{theorem}
    \label{thm:treewidth}
    Given a tree decomposition of width $\tw$ of the underlying undirected dataflow graph,
    \prob{Layout Selection} can be solved optimally in time
    $O\!\left(|V|\cdot \tw^2 \cdot M^{\tw+2}\cdot 4^{M(\tw+1)}\right)$, where $M$ bounds the number of
    layout configurations per operator. In particular, the problem is fixed-parameter tractable in
    $(\tw,M)$, with runtime linear in the number of operators for fixed $\tw$ and $M$.
\end{theorem}

\begin{figure}[!t]
    \centering
    \begin{subfigure}[b]{.4\linewidth}
        \centering
        \includegraphics[page=3,height=6cm]{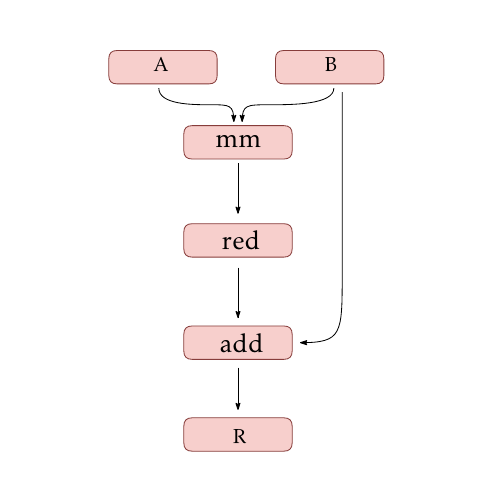}
        \caption{Dataflow graph}
    \end{subfigure}
    \hfill
    \begin{subfigure}[b]{.55\linewidth}
        \centering
        \includegraphics[page=4,height=7cm]{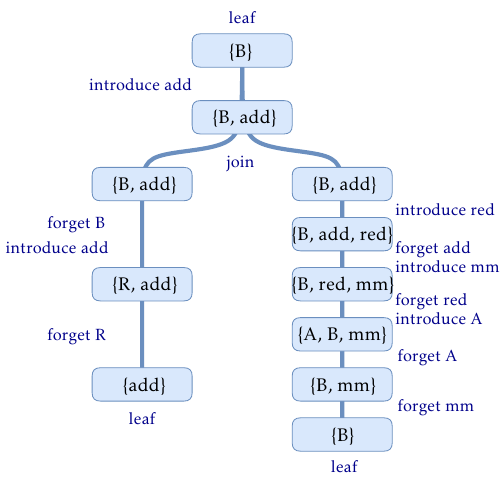}
        \caption{Nice tree decomposition of width $2$}
    \end{subfigure}
    \caption{Treewidth-based layout selection applied to the dataflow graph from
        Figure~\ref{fig:layout-example}. Each node in the tree decomposition is labeled by type
        (leaf, introduce, forget, or join), and annotated with its bag contents. The algorithm
        of Theorem~\ref{thm:treewidth} traverses this decomposition bottom-up: introduce nodes
        record operator configurations, forget nodes charge operator and transpose costs, and the
        join node combines two subtrees.}
    \Description[A four-operator dataflow graph and a nice tree decomposition of it with bag
    contents]{Panel~(a) redraws the dataflow graph of Figure~\ref{fig:layout-example} with
        abbreviated names: input tensors $A$ and $B$ feed the matmul \texttt{mm}, whose result
        feeds the reduction \texttt{red}, whose result feeds \texttt{add}; $B$ also feeds
        \texttt{add} directly by a long edge on the right, and \texttt{add} produces $R$.
        Panel~(b) draws a nice tree decomposition of width~$2$, so every bag holds at most three
        elements. The root is a leaf-labelled bag $\{B\}$; below it an ``introduce \texttt{add}''
        step gives $\{B,\texttt{add}\}$, which is a join node with two children. The left branch
        goes $\{B,\texttt{add}\}$, then ``forget $B$, introduce \texttt{add}'' to $\{R,
        \texttt{add}\}$, then ``forget $R$'' to the leaf $\{\texttt{add}\}$. The right branch goes
        $\{B,\texttt{add}\}$, then ``introduce \texttt{red}'' to $\{B,\texttt{add},\texttt{red}\}$,
        then ``forget \texttt{add}, introduce \texttt{mm}'' to $\{B,\texttt{red},\texttt{mm}\}$,
        then ``forget \texttt{red}, introduce $A$'' to $\{A,B,\texttt{mm}\}$, then ``forget $A$''
        to $\{B,\texttt{mm}\}$, then ``forget \texttt{mm}'' to the leaf $\{B\}$.}
    \label{fig:treewidth}
\end{figure}

\paragraph{The algorithm.}
The algorithm sweeps a nice tree decomposition of the underlying undirected dataflow graph
bottom-up, from the leaves to the root, maintaining a table of partial solutions at every tree node.
A table entry at node $i$ summarizes one way of laying out the part of the graph already processed
(the operators that occur in the subtree below $i$) and is indexed by two pieces of information:
\begin{itemize}
    \item a configuration $f(o)\in\St(o)$ for every operator $o$ in the current bag $X_i$,
    fixing the layouts of all tensors that $o$ touches, and
    \item for every tensor $t$ produced by a bag operator, the set $S_t\subseteq L(t)$ of
    layouts that already-processed consumers of $t$ have requested.
\end{itemize}
The value of the entry is the minimum cost attainable by the processed part of the graph, consistent
with that index. The four node types of the nice decomposition translate into four transition rules.
At a \emph{leaf}, the table is trivial (empty index, cost $0$). At an \emph{introduce} node,
operator $v$ enters the frontier: its candidate configurations are enumerated and recorded in the
index, but no cost is charged yet, and the consumer-request set of its output tensor starts out
empty. At a \emph{forget} node, operator $v$ disappears from the frontier and never reappears above
$i$. Three things
happen: $v$'s operator cost is charged; $v$'s input-layout requirements are registered in the
request sets of its producers that are still in the bag; and the transpose costs of $v$'s output
tensor (once per \emph{distinct} requested layout) are charged, since by now every consumer of that
tensor has either been forgotten (its request is recorded in the request set) or still resides in
the bag (its request is visible in the bag configuration), so the tensor's full set of consumer
layouts is known. Finally, at a \emph{join} node, the tables of two independently processed subtrees
with identical bags are combined: entries with the same bag configuration are merged, and their
request sets are united point-wise, so a layout requested in both subtrees is charged only once. At
the root, the bag is empty and the single remaining entry is the global optimum; the optimal
assignment is recovered by standard backtracking from the root to the leaves.

Correctness follows from the connectedness property of tree decompositions. A producer shares a bag
with each of its consumers, so every consumer is introduced below the forget node. When the transpose
cost is charged, the layouts requested by \emph{all} consumers are therefore known. The runtime is
dominated by join nodes, which combine pairs of request sets. The complete
dynamic program, with formal transition rules and the correctness and runtime analysis, is given in
Appendix~\ref{app:treewidth-proof}.

In Figure~\ref{fig:treewidth}, for instance, the configuration of the matrix multiplication
(\texttt{mm}) is recorded at the introduce node containing it; its operator cost and the transpose
costs of its output tensor are charged at the forget node where it leaves the frontier, by which
point all of its consumers have registered their layout requests; and the join node combines the
independently processed branches of the graph.

\subsubsection{MaxSAT-based Encoding}
\label{sec:maxsat}

Our second exact algorithm encodes \prob{Layout Selection} as a weighted \prob{MaxSAT}
instance~\cite{DBLP:series/faia/LiM21}. Since the runtime and the memory usage of the exact dynamic programming
(Theorem~\ref{thm:treewidth}) scale exponentially in the width, the algorithm quickly exhausts the
memory budget on wider graphs. The \prob{MaxSAT} encoding provides an alternative that does not rely on a tree
decomposition.

\paragraph{Operator variables and clauses.}
For each operator $o \in V$, we introduce a Boolean variable
$x_{o,\vec{\ell}_{\text{in}},\ell_{\text{out}}}$ for each \df{configuration}
$\vec{\ell}_{\text{in}} \in L(\In(o))$ and $\ell_{\text{out}} \in L(\Out(o))$, which is true if and
only if $o$ uses input layouts $\vec{\ell}_{\text{in}}$ and output layout $\ell_{\text{out}}$.

The following hard clauses enforce that each operator selects exactly one feasible configuration.
Every operator must pick at least one configuration, and no two distinct configurations may be
active simultaneously:
\begin{align*}
    \bigvee_{\vec{\ell}_{\text{in}}, \ell_{\text{out}}}
    x_{o,\vec{\ell}_{\text{in}},\ell_{\text{out}}}
    &\quad \forall o \in V, \\
    \neg x_{o,\vec{\ell}_{\text{in}},\ell_{\text{out}}}
    \vee \neg x_{o,\vec{\ell}'_{\text{in}},\ell'_{\text{out}}}
    &\quad \forall o \in V,\;
    (\vec{\ell}_{\text{in}}, \ell_{\text{out}}) \neq (\vec{\ell}'_{\text{in}},
    \ell'_{\text{out}}).
\end{align*}

Infeasible configurations, those with $\opcost_o(\vec{\ell}_{\text{in}},\ell_{\text{out}}) = \infty$,
receive no variable at all, so they cannot be selected; this keeps the encoding polynomial in the
number of feasible configurations and the graph size. Equivalently one could introduce the variable
and exclude it with a unit clause $\neg x_{o,\vec{\ell}_{\text{in}},\ell_{\text{out}}}$.

Soft clauses penalize operator costs. For each $o \in V$ and feasible configuration
$(\vec{\ell}_{\text{in}}, \ell_{\text{out}})$:
\begin{equation*}
    \neg x_{o,\vec{\ell}_{\text{in}},\ell_{\text{out}}}
    \quad \text{with weight } \opcost_o(\vec{\ell}_{\text{in}}, \ell_{\text{out}}).
\end{equation*}

\paragraph{Transpose variables and clauses.}
For each tensor $t \in \mathcal{T}$ and layout $\ell \in L(t)$, we introduce a variable
$y_{t,\ell}$, which is true if some consumer of $t$ requires layout $\ell$.

Hard clauses link operator configurations to required tensor layouts. If operator $o$ consumes
tensor $t$ in layout $\ell$ under configuration $c$:
\begin{equation*}
    \lnot x_{o,c} \vee y_{t,\ell}.
\end{equation*}

Soft clauses penalize transpose costs. For each tensor $t \in \mathcal{T}$ and its unique producing
operator $o_{\text{prod}}$, we add a clause for each configuration $c$ of $o_{\text{prod}}$ that
outputs $t$ in layout $\ell$, and each distinct layout $\ell' \neq \ell$:
\begin{equation*}
    \neg x_{o_{\text{prod}},c} \vee \neg y_{t,\ell'}
    \quad \text{with weight } \tcost_t(\ell, \ell').
\end{equation*}

This clause is violated when tensor $t$ is produced in layout $\ell$ (because its producing operator
$o_{\text{prod}}$ selects configuration $c$) but consumed in layout $\ell'$. The cost is charged
once per distinct consumer layout, regardless of how many operators share that layout. The total
weight of the violated soft clauses therefore equals the objective of \prob{Layout Selection}, so a
weighted \prob{MaxSAT} solver minimizing that weight returns an optimal layout assignment.

\section{Experimental Evaluation}
\label{sec:evaluation}

We evaluate our layout selection model on a specialized AI hardware accelerator, AWS
Trainium~\cite{aws_trainium}. The accelerator uses an established production
compiler~\cite{neuron_compiler} to compile ML models from various frameworks (TensorFlow, PyTorch,
XLA HLO) into optimized code; see Appendix~\ref{app:compiler} for a high-level overview. The
experiments address four research questions:

\begin{enumerate}
    \item[\textbf{RQ1}] (\emph{Practical Impact of Layout Selection}) Is sophisticated layout
    selection beneficial, and how much execution-time performance is sacrificed by adopting a fast
    heuristic over an optimized approach?

    \item[\textbf{RQ2}] (\emph{Solver-based Algorithms in Practice}) Are exact methods
    (\alg{MaxSAT}, \alg{Treewidth}) viable within a production compiler? What is the empirical
    treewidth of real ML dataflow graphs, and at what scale do these methods become impractical?

    \item[\textbf{RQ3}] (\emph{Algorithm Comparison}) How do the studied layout selection
    algorithms compare on formal objective value, on resulting model execution time, and on
    compilation time?

    \item[\textbf{RQ4}] (\emph{Cost Model Fidelity}) How well does the formal objective predict
    actual execution time? In particular, can we reliably rank algorithms based on estimated costs
    alone?

\end{enumerate}

\subsection{Experimental Setup}
\label{sec:experimental_setup}

We consider four different layout selection strategies within our compiler setup:
\begin{description}
    \item[\alg{Local}{\normalfont :}] For each operator independently, select the layout
    configuration with minimum operator cost, ignoring transpose costs entirely. This baseline
    isolates the contribution of operator costs alone.

    \item[\alg{Greedy}{\normalfont :}] First, an initial assignment is constructed by processing
    operators in topological order, selecting for each the configuration minimizing its operator
    cost plus transpose costs to already-assigned predecessors. The assignment is then refined
    iteratively, re-optimizing each operator's layout choice against the total objective (operator
    costs plus transpose costs) with all other assignments held fixed, until no single-operator
    change yields an improvement. This mimics XLA's greedy heuristic~\cite{xla_layout_impl}.

    \item[\alg{Rule-Based}{\normalfont :}] We first group tensor dimensions that refer to the same
    logical dimension across operators. Starting from operators whose preferred layout is known
    (such as matrix multiplications) we propagate layout constraints to neighboring operators. We
    then split the graph into small components, apply domain-specific rules inside each component
    (for example, preferring layouts that avoid known-expensive instruction patterns), and search
    the remaining choices exhaustively. This is the strategy the AWS Neuron Graph Compiler ships
    with~\cite{neuron_compiler}.

    \item[\alg{Solver-Based}{\normalfont :}] We compute an optimal solution minimizing the overall
    costs in two ways: \alg{Treewidth} runs the dynamic program of Section~\ref{sec:treewidth}, and
    \alg{MaxSAT} solves the weighted \prob{MaxSAT} encoding of Section~\ref{sec:maxsat} with an
    off-the-shelf solver. Below we write \prob{MaxSAT} for the problem and \alg{MaxSAT} for the
    strategy that solves our encoding of it.
\end{description}

We evaluate on open-source models from HuggingFace spanning diverse architectural families:
decoder-only LLMs (OLMo, Qwen, Ministral, VaultGemma), encoder-only and encoder-decoder
transformers (BERT, ALBERT, Electra, XLM-RoBERTa, BART, DistilBART), CNN-based vision models (ResNet,
CNN), diffusion models, mixture-of-experts models (OLMoE, FlexOlmo, Granite), multimodal architectures
(Perceiver, Kosmos, PaLiGemma, Qwen3-VL), and audio models (Whisper, Qwen2-Audio, MusicGen), covering both sequence- and spatial-dominated workloads.
By varying model parameters (hidden dimensions, number of
layers) and context lengths, we obtain a benchmark suite of over $200$ model configurations ranging
from $43$ to $10{,}056$ operators. An operator here is a vertex of the dataflow graph, such as
a matmul or a convolution; the compiler later expands one operator into many hardware instructions. We report aggregate statistics across the full suite and present
detailed measurements for a representative subset of $24$ configurations. For each architecture
family, we selected the largest configuration whose compiled execution is
dominated by the model itself rather than by fixed overheads, excluding trivially small variants
(which all algorithms handle equally well) and configurations exceeding the memory budget of our
compilation environment.

We track two primary metrics: \emph{model performance} (end-to-end execution time of a compiled
model processing an inference request in hardware cycles) and \emph{algorithm efficiency}
(wall-clock compilation time and peak memory consumption of the layout selection pass). The
experiments are conducted on an Amazon EC2 \texttt{trn1.32xlarge} instance, which provides 16
Trainium chips and 128 vCPUs of Intel Xeon Platinum 8375C processors (2.90\,GHz, 54\,MB L3 cache);
each compiled model is executed on a single Trainium chip of that instance, and the layout selection
pass runs on the host CPUs.

Two aspects of the setup matter for interpreting the results. First, all four strategies run at the
same (early) position in the Neuron compilation pipeline: operator fusion, tiling, instruction
scheduling, and memory allocation all execute \emph{after} layout selection. The cost model
therefore estimates costs on the pre-fusion operator graph, and the same downstream passes are
applied to the output of every strategy, so the comparison is not confounded by differing fusion
decisions.
Second, \alg{Local}, \alg{Greedy}, and \alg{Solver-Based} use the same cost model: \alg{Local}
minimizes only operator costs, \alg{Greedy} locally improves the full objective, and
\alg{Solver-Based} minimizes it exactly. \alg{Rule-Based} instead follows domain rules, though we
evaluate the objective on its output.

\subsection{Results}

\noindent\textbf{RQ1} (\emph{Practical Impact of Layout Selection}).

To quantify the value of layout selection, we compare the production \alg{Rule-Based} algorithm
against the two baseline approaches, \alg{Local} and \alg{Greedy}. Figure~\ref{rq1a} summarizes the
results across our full benchmark suite by plotting the relative model performance (as a fraction of
the best known result) achieved by each algorithm, with models sorted from best to worst. The
\emph{best known result} for a model is the minimum hardware execution time across all four
evaluated algorithms. A value of $100$ indicates that the algorithm matches the best known
performance for that model; lower values indicate worse model execution time.

\begin{figure}[!tb]
    \centering
    \begin{subfigure}[b]{.49\linewidth}
        \centering
        \includegraphics[width=\textwidth]{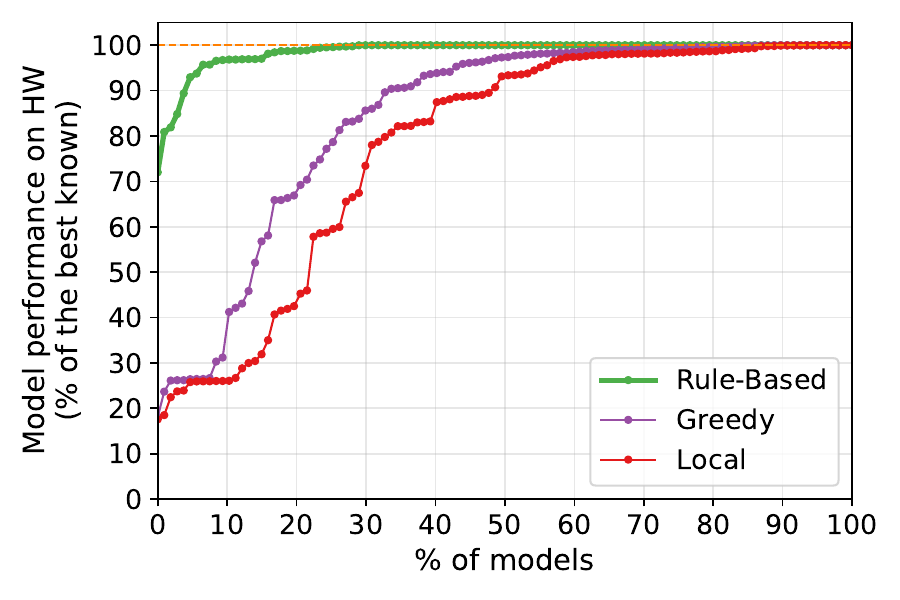}
        \caption{All models}
        \label{rq1a}
    \end{subfigure}
    \hfill
    \begin{subfigure}[b]{.49\linewidth}
        \centering
        \includegraphics[width=\textwidth]{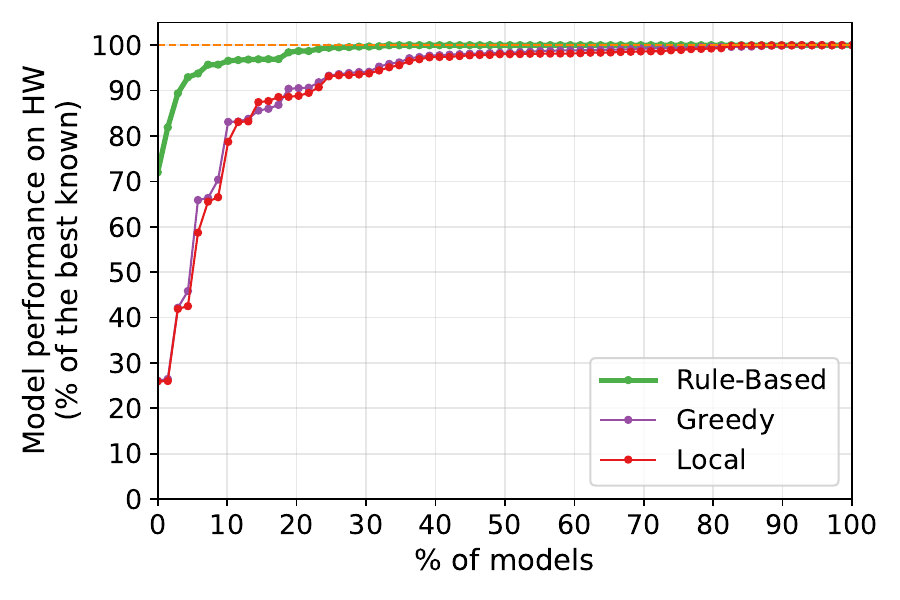}
        \caption{Dense Transformers}
        \label{rq1b}
    \end{subfigure}
    \caption{\textbf{RQ1} (Practical Impact of Layout Selection): Relative model performance (as
        fraction of best known result) on Trainium instances.}
    \Description[Two cumulative-distribution plots of relative model performance, over all models
    and over dense transformers only]{Both panels plot model performance on hardware as a
        percentage of the best known result ($y$, $0$ to $100$) against the percentage of models
        ($x$, $0$ to $100$), as one non-decreasing curve per algorithm: \alg{Rule-Based} in green,
        \alg{Greedy} in purple and \alg{Local} in red, with a dashed horizontal line at $100\%$.
        Higher and earlier curves are better. In panel~(a), over all models, the green
        \alg{Rule-Based} curve starts near $71\%$ and rises steeply, exceeding $95\%$ within the
        first few percent of models and staying essentially flat at $100\%$ thereafter. The purple
        \alg{Greedy} and red \alg{Local} curves start below $20\%$, remain under $30\%$ for the
        first tenth of models, then climb gradually, with \alg{Greedy} consistently above
        \alg{Local} by roughly ten percentage points through the middle of the range; both
        converge to the green curve only past about $60\%$ of models. Panel~(b), restricted to
        dense transformers, shows the same green curve but a much smaller spread: \alg{Greedy} and
        \alg{Local} nearly coincide, rise past $80\%$ within the first tenth of models, and track
        just below \alg{Rule-Based} for the rest of the range.}
    \label{rq1}
\end{figure}

Compared to \alg{Rule-Based}, which is the default strategy in the compiler,
the simpler strategies lose about $20\%$ of performance on average, and in the worst decile they
reach only one fifth of the best known result, corresponding to a $5\times$ slowdown. \alg{Greedy} improves over \alg{Local}, but the difference is
modest. We found that the severity of the regression depends on the model family. For regular dense
transformer architectures (such as BERT, ALBERT, Electra, BART), even the simplest \alg{Local}
baseline, which
ignores transpose costs entirely, stays within $5\%$ of the best result on most such models; see
Figure~\ref{rq1b}. Complex vision models behave differently: \alg{Local} takes 96.0
million cycles on ResNet-50 versus 20.9 for \alg{Rule-Based} (a $4.6\times$ slowdown) and \alg{Greedy}
fares only slightly better at 73.5 million cycles; refer to Table~\ref{tab:results} for the
measurements.

The root cause is an increase in inserted transpose operations: $20\%$ more on average over the
benchmark suite, rising to $2\times$ at the 90th percentile (p90).
The conversions are expensive for the reasons discussed in Section~\ref{sec:conversions}: on
Trainium, a partition-dimension change physically moves every element across on-chip memory partitions through
dedicated hardware engines and, unlike, e.g., the \texttt{op(A)}/\texttt{op(B)} operand flags of
cuBLAS kernels on GPUs, cannot be fused into the consuming operator. In addition to directly
increasing model execution time on hardware, the extra conversions complicate subsequent compiler
passes in the Trainium compilation pipeline. In particular, the tiling pass, which runs after layout
selection, decomposes each operator into hardware instructions by splitting large tensors into tiles
that fit in on-chip memory. Each tile requires a separate DMA transfer between device memory and the
scratchpad buffer, and each transfer incurs a fixed setup overhead. Layout conversions inserted
between operators force the tiler to generate many more small, strided DMA transfers than a
consistent layout would require. This
explains the regression on ResNet-50: \alg{Local} produces $24\times$ higher estimated DMA
cost after tiling compared to \alg{Rule-Based}.

\medskip
\noindent\textbf{RQ2} (\emph{Solver-based Algorithms in Practice}).

The treewidth-based algorithm runs in time polynomial in graph size but exponential in treewidth;
see Theorem~\ref{thm:treewidth}. To assess its practicality, we first estimate the treewidth of
dataflow graphs extracted from models in our benchmark suite. Figure~\ref{tw1} reports the results:
while $92\%$ of instances have treewidth not exceeding $4$, occasionally we encounter an
instance with larger treewidth (the maximum is $8$ in our evaluation). This observation aligns with
prior observations that ML dataflow graphs have bounded treewidth~\cite{0001PSVW19}. However, the
dynamic program of Theorem~\ref{thm:treewidth} can build tables too large to fit in memory, making
the approach impractical for large instances. This is
confirmed by Figure~\ref{tw2}, where we plot the runtime (in log scale) needed to solve a given
fraction of benchmark instances. The dynamic programming finds an optimal solution to \prob{Layout Selection} within a few seconds for graphs with treewidth at most $4$, covering the same
$92\%$ of our benchmark suite. However, larger instances require substantially more memory (over
$10$\,GB for treewidth-$5$ graphs and over $64$\,GB for larger ones), which exceeds the memory
budget of a typical compilation environment. Unlike \alg{MaxSAT}, the DP algorithm is
all-or-nothing, in that it either completes with an optimal solution or produces no result at all.

\begin{figure}[!tb]
    \centering
    \begin{subfigure}[b]{.5\linewidth}
        \centering
        \includegraphics[width=\textwidth]{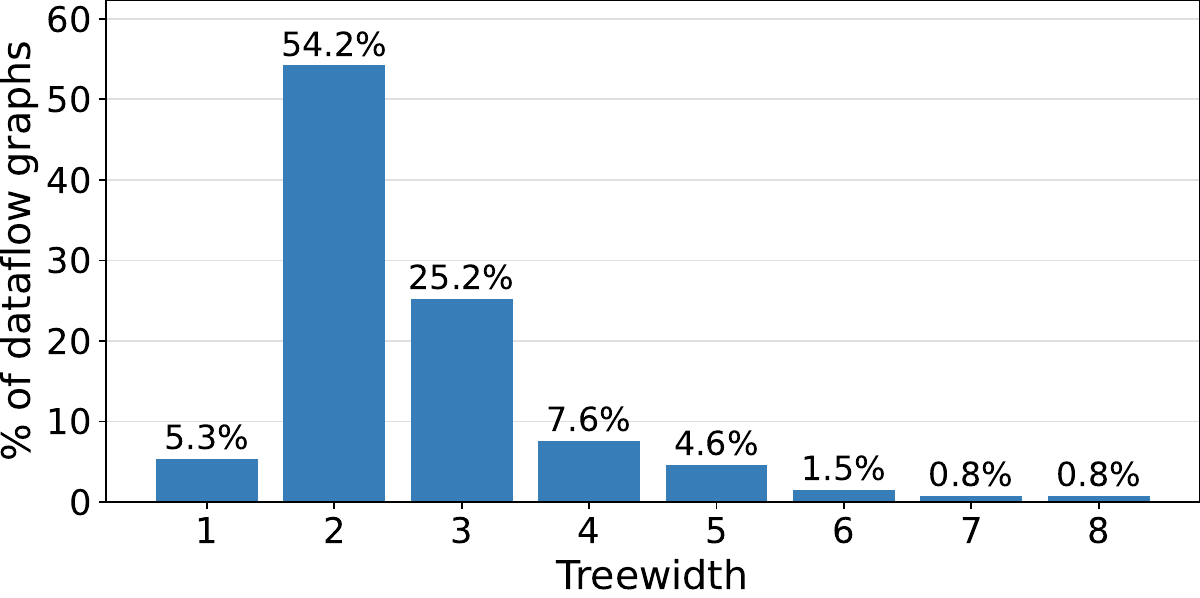}
        \caption{Treewidth of dataflow graphs}
        \label{tw1}
    \end{subfigure}
    \hfill
    \begin{subfigure}[b]{.43\linewidth}
        \centering
        \includegraphics[width=\textwidth]{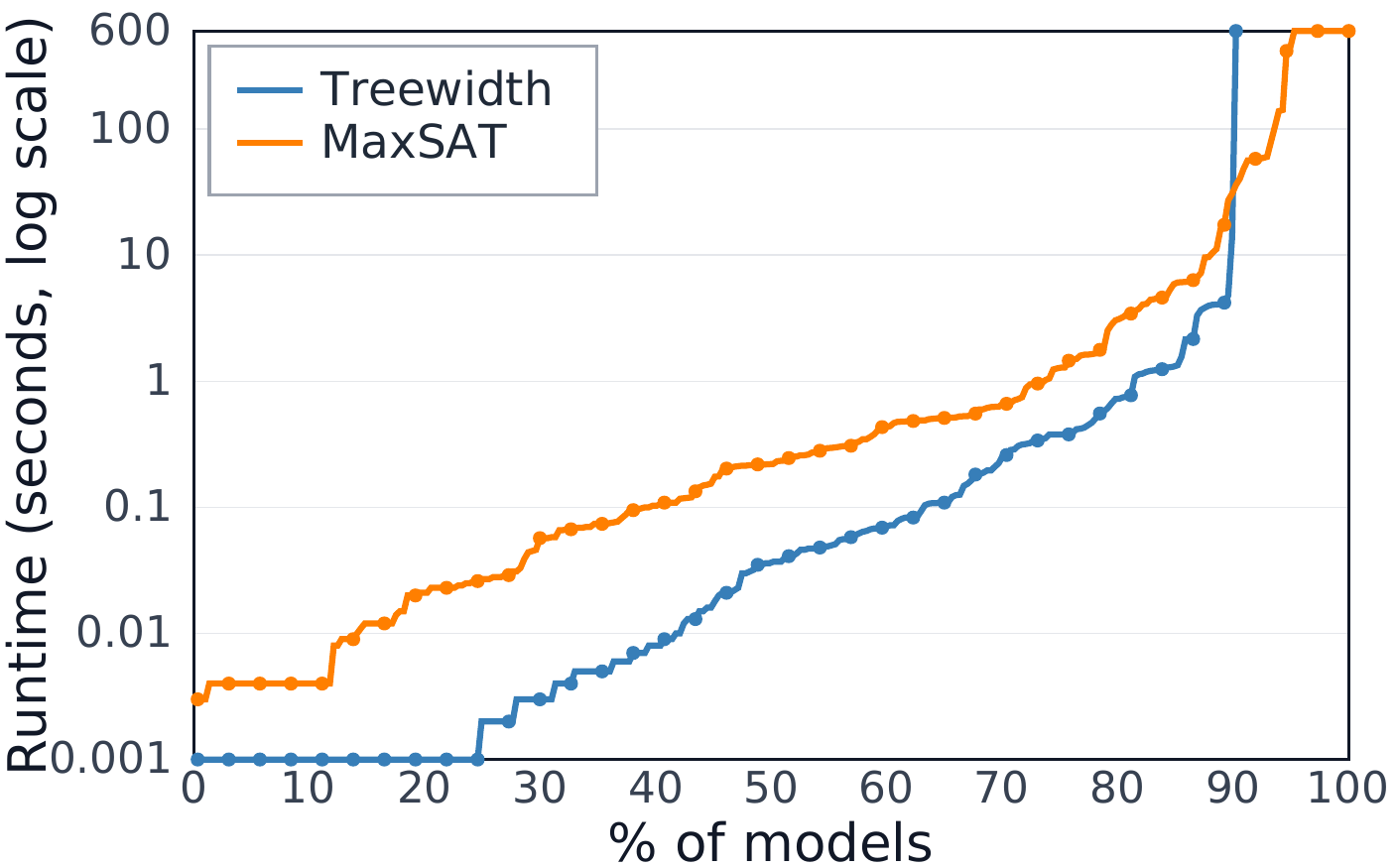}
        \caption{Runtime of \alg{Treewidth} and \alg{MaxSAT}}
        \label{tw2}
    \end{subfigure}
    \caption{\textbf{RQ2} (Solver-based Algorithms in Practice): Most real-world dataflow graphs
        exhibit low treewidth, yet \alg{Treewidth} exceeds the memory budget on $8\%$ of
        instances, whereas \alg{MaxSAT} remains practical.}
    \Description[A bar chart of the treewidth distribution and a cumulative runtime plot for the
    two exact algorithms]{Panel~(a) is a bar chart of the percentage of dataflow graphs
        ($y$, $0$ to $60$) by treewidth ($x$, $1$ to $8$), with each bar labelled: $5.3\%$ at
        width~$1$, $54.2\%$ at~$2$, $25.2\%$ at~$3$, $7.6\%$ at~$4$, $4.6\%$ at~$5$, $1.5\%$
        at~$6$, and $0.8\%$ at each of~$7$ and~$8$; the distribution is thus concentrated on
        widths $2$ and~$3$, which together account for about four fifths of the graphs. Panel~(b)
        plots runtime in seconds on a logarithmic $y$~axis from $0.001$ to $600$ against the
        percentage of models ($x$, $0$ to $100$), with one non-decreasing curve for
        \alg{Treewidth} in blue and one for \alg{MaxSAT} in orange. The blue \alg{Treewidth} curve
        sits roughly one order of magnitude below the orange \alg{MaxSAT} curve across the whole
        middle of the range, staying under a tenth of a second for about $70\%$ of models, but
        then jumps vertically to the $600$-second ceiling at about $90\%$ of models, marking the
        instances that exhaust the memory budget. The orange \alg{MaxSAT} curve rises more
        smoothly, stays below ten seconds for about $90\%$ of models, and reaches the same ceiling
        only in the last few percent.}
    \label{tw}
\end{figure}

Figure~\ref{tw2} also illustrates runtimes of \alg{MaxSAT} for the considered models. We solve the
\alg{MaxSAT} encoding using Z3's optimization module~\cite{Z3} with a 10-minute timeout. For
transformer models, \alg{MaxSAT} is viable. It completes within seconds on most configurations, at an
additional compilation overhead under $10$ seconds, usually less than 5\% of the total model
compilation time. In the rare cases when the solver times out (observed only on the
large models, e.g., Vision Perceiver with $3{,}205$ operators), the best feasible solution found
within the first $10$ seconds is not improved during the remaining time. The solver spends the bulk
of its budget proving optimality rather than finding better assignments, which suggests the partial
solutions returned at timeout are close to optimal, though we cannot bound the gap. We therefore use \alg{MaxSAT} as the primary
exact method in the following sections.

The usefulness of an exact solver rests on the cost model being accurate enough that improvements in
the formal objective show up as improvements on hardware. Where that fails, provably optimal
solutions to \prob{Layout Selection} can underperform simpler strategies, the gap we examine in
\textbf{RQ3} and \textbf{RQ4}.

\definecolor{collocal}{rgb}{0.894, 0.102, 0.110}
\definecolor{colgreedy}{rgb}{0.596, 0.306, 0.639}
\definecolor{colrule}{rgb}{0.302, 0.686, 0.290}
\definecolor{colsolver}{rgb}{0.216, 0.494, 0.722}

\begin{table}[!t]
    \centering
    \small
    \caption{Model execution time (millions of cycles). Lower is better. The best result per model,
        and every result within $1\%$ of it, is bold.}
    \label{tab:results}

    \begin{tabular*}{\textwidth}{@{\extracolsep{\fill}} l l r r r r}
        \toprule
        Type & Model & \colorbox{collocal!50}{\alg{Local}} &
        \colorbox{colgreedy!50}{\alg{Greedy}} &
        \colorbox{colrule!50}{\alg{Rule-Based}} &
        \colorbox{colsolver!50}{\alg{Solver-Based}}
        \\
        \midrule

        \multirow{12}{*}{\rotatebox[origin=c]{90}{Dense Transformers}}
        & ALBERT-large        &   53.4 &   45.8 & \textbf{ 44.7} &   45.3 \\
        & BERT-base           &   19.2 &   19.3 & \textbf{ 18.5} & \textbf{ 18.5} \\
        & BERT-large          &   29.0 &   29.6 & \textbf{ 28.7} & \textbf{ 28.7} \\
        & BART-base           & \textbf{  2.4} &    2.6 & \textbf{  2.4} & \textbf{  2.4} \\
        & DistilBART          &    2.2 &    2.3 &    2.2 & \textbf{  2.1} \\
        & Electra-large       & \textbf{ 23.2} &   23.7 & \textbf{ 23.0} & \textbf{ 23.0} \\
        & Ministral-4B        &    3.5 &    3.5 & \textbf{  3.3} &    3.5 \\
        & OLMo-7B             &   51.1 &   51.4 & \textbf{ 50.1} &   51.6 \\
        & VaultGemma-1B       &   13.8 & \textbf{ 10.2} & \textbf{ 10.1} & \textbf{ 10.1} \\
        & Whisper enc         & \textbf{225} & \textbf{226} &  258 & \textbf{225} \\
        & Whisper dec         &    7.2 & \textbf{  7.1} &    7.2 &    7.2 \\
        & XLM-RoBERTa         &    3.1 &    3.1 & \textbf{  2.9} & \textbf{  2.9} \\

        \midrule

        \multirow[c]{12}{*}{\rotatebox[origin=c]{90}{
                \begin{tabular}{c}
                    Vision,\\
                    Multimodal, MoE
        \end{tabular}}}
        & CNN                 &    1.8 &    1.1 & \textbf{  0.9} &    1.1 \\
        & FlexOlmo-7x7B       &  146 & \textbf{138} & \textbf{139} &  154 \\
        & Granite-3.1-3B      & \textbf{184} &  188 &  198 & \textbf{185} \\
        & Kosmos-2.5          &   17.7 &   23.0 &   19.2 & \textbf{ 17.3} \\
        & Language Perceiver  &  462 &  480 & \textbf{346} &  516 \\
        & MusicGen            &   44.3 & \textbf{ 27.0} & \textbf{ 27.0} &   38.5 \\
        & OLMoE-1B-7B         &  100 &  109 & \textbf{ 96.5} & \textbf{ 96.9} \\
        & PaLiGemma-3B        & \textbf{ 10.8} &   11.0 &   11.7 &   13.3 \\
        & Qwen2-Audio         &    9.9 &    9.9 & \textbf{  9.6} &    9.9 \\
        & Qwen3-VL            &   34.9 &   31.6 & \textbf{ 30.7} &   35.9 \\
        & ResNet-50           &   96.0 &   73.5 & \textbf{ 20.9} &   63.9 \\
        & Vision Perceiver    &   14.1 &   14.8 & \textbf{ 10.2} &   11.3 \\

        \bottomrule
    \end{tabular*}
\end{table}

\medskip
\noindent\textbf{RQ3} (\emph{Algorithm Comparison}).

We first compare the algorithms with respect to the formal objective used in layout selection.
Figure~\ref{fig:gap} reports the optimality gap of each algorithm relative to the
\alg{Solver-Based} optimum ($0\%$ = optimal). \alg{Greedy} produces solutions
typically within $1\%$ of optimal, while
the \alg{Rule-Based} heuristic is within $3\%$ on average ($5\%$ at p90).
Despite these small objective differences, the resulting hardware performance can differ
significantly.

Table~\ref{tab:results} summarizes execution cycles for each algorithm across a representative
subset of our benchmark suite. We categorize models into two groups. \emph{Dense transformers}
(BERT, ALBERT, Electra, BART, OLMo, Whisper, etc.) have regular data flow where layout preferences
are mostly consistent across operators, resulting in few layout conflicts. \emph{Complex
    architectures} (vision models, multimodal, mixture-of-experts) contain heterogeneous operators
(convolutions, cross-attention, modality encoders, expert routing) along with skip connections and
varying tensor shapes, creating many layout conflicts between adjacent operators. We note that the
measurement noise (execution time variation across repeated runs) is below
$0.5\%$ in our experiments; the differences discussed below all exceed this margin.

For dense transformers, the four algorithms perform comparably. \alg{Local} and \alg{Solver-Based}
frequently match \alg{Rule-Based}: all three tie for the best execution time on Electra-large and
BART-base, and \alg{Solver-Based} also matches it on BERT-base and BERT-large. This confirms the
finding from \textbf{RQ1} that for regular architectures with uniform layout preferences, even
ignoring transpose costs often yields
near-optimal results. \alg{Rule-Based} is $15\%$ slower only on Whisper enc (258 vs.\ 225);
elsewhere it wins by small margins (e.g., Ministral-4B, OLMo-7B), likely due to domain-specific
heuristics tuned for these decoder-only architectures.

Complex architectures behave differently. \alg{Rule-Based} is the winner on most models,
sometimes by a wide margin: ResNet-50 (20.9 vs.\ 96.0 for \alg{Local}), MusicGen
(27.0 vs.\ 44.3), Vision Perceiver (10.2 vs.\ 14.1), and Qwen3-VL (30.7 vs.\ 34.9). \alg{Greedy}
often improves over \alg{Local} but still falls far short of \alg{Rule-Based}. More surprisingly,
\alg{Solver-Based}, which computes provably optimal solutions under the cost model, sometimes
performs \emph{worse} than simple heuristics, as on ResNet-50 (63.9 vs.\ 20.9 for \alg{Rule-Based}),
Language Perceiver (516 vs.\ 346), and PaLiGemma-3B (13.3 vs.\ 10.8 for \alg{Local}). Minimizing the
formal objective evidently does not guarantee good hardware performance when the cost model is
inaccurate.

\alg{Rule-Based} produces the best results on complex architectures and competitive
results on dense transformers. The product of multi-year development, it encodes domain-specific
knowledge (e.g., preferred partition dimensions for convolutions) that the cost model does not
capture; Section~\ref{sect:analysis} dissects why such a heuristic can beat a provably optimal
algorithm.

\medskip
\noindent\textbf{RQ4} (\emph{Cost Model Fidelity}).

\begin{figure}[!tb]
    \centering
    \begin{subfigure}[b]{.5\linewidth}
        \centering
        \includegraphics[width=\textwidth]{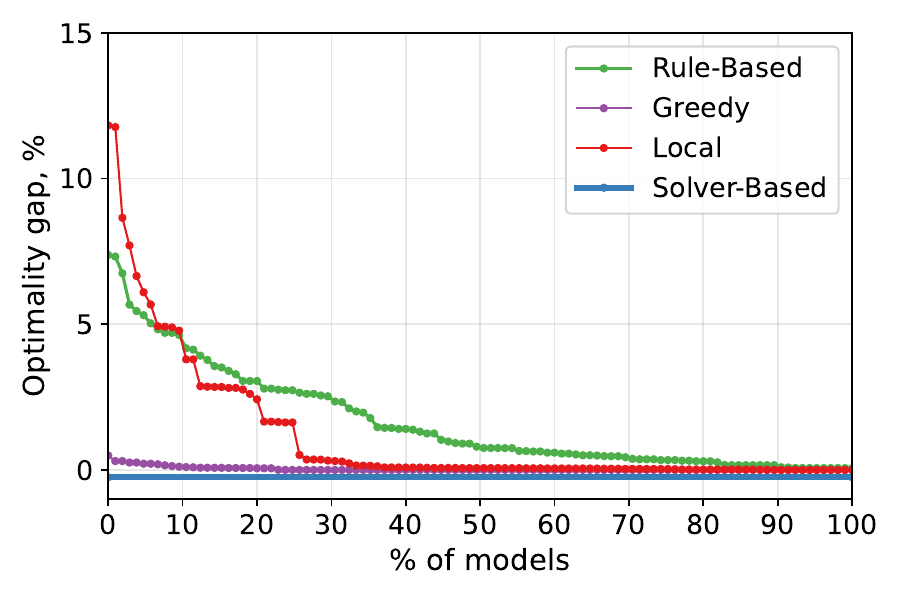}
        \caption{Optimality gap relative to the \alg{Solver-Based} optimum ($0\%$ = optimal)
            under the cost model}
        \label{fig:gap}
    \end{subfigure}
    \hfill
    \begin{subfigure}[b]{.4\linewidth}
        \centering
        \includegraphics[width=\textwidth]{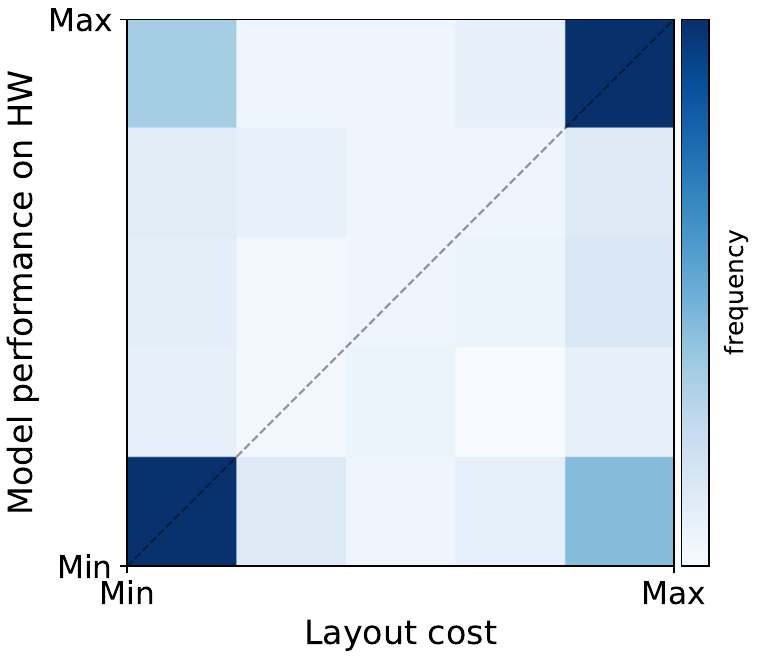}
        \caption{Agreement between cost-model predictions and hardware execution time}
        \label{fig:fidelity}
    \end{subfigure}
    \caption{\textbf{RQ4} (Cost Model Fidelity): Algorithms achieve near-optimal
        objective values under the cost model,
        yet predicted costs can diverge from actual hardware execution time.}
    \Description[A cumulative plot of optimality gaps and a heatmap of predicted cost against
    measured performance]{Panel~(a) plots the optimality gap in percent ($y$, $0$ to $15$)
        against the percentage of models ($x$, $0$ to $100$), as one non-increasing curve per
        algorithm: \alg{Solver-Based} in blue lies flat on $0\%$ by definition; \alg{Greedy} in
        purple starts at about $0.5\%$ and is indistinguishable from zero past the first tenth of
        models; \alg{Local} in red starts near $12\%$ and falls below $1\%$ by about a quarter of
        models; \alg{Rule-Based} in green starts near $7\%$ and decays most slowly of all, staying
        above the other heuristics from roughly $30\%$ of models onward. All curves are therefore
        within a few percent of the optimum over most of the suite. Panel~(b) is a
        $5\times 5$ heatmap whose $x$~axis is the normalized predicted layout cost from Min to Max
        and whose $y$~axis is the normalized measured model performance on hardware from Min to
        Max, with cell shading giving the frequency of algorithm-model points and a dashed
        diagonal marking perfect agreement. Two dark cells sit on the diagonal, at the
        lowest-cost/best-performance corner and the highest-cost/worst-performance corner, but two
        further clusters of mass sit off it: a moderately dark cell at lowest predicted cost with
        worst measured performance, and a shaded cell at highest predicted cost with best measured
        performance. The interior cells are all pale, so the distribution is concentrated in the
        corners rather than along the diagonal.}
    \label{fig:rq4}
\end{figure}

To evaluate cost model accuracy, we use a pairwise comparison methodology. For each benchmark model
and each pair of layout selection algorithms, we compare the predicted layout cost (operator costs
plus transpose costs) and check whether the algorithm with the lower predicted cost also achieves
lower hardware execution time. We define accuracy as the fraction of comparisons where the cost
model's ranking agrees with the hardware outcome. This metric is related to Kendall's rank
correlation coefficient~\cite{kendall1938}, but applied to algorithm pairs rather than full
rankings. To avoid disagreements caused by measurement noise, we treat two values as equal when
their relative difference is below $0.5\%$.

In our evaluation, the cost model correctly ranks approximately $87\%$ of pairwise comparisons;
in contrast, a random predictor, which would guess uniformly among the three possible outcomes
(``better,'' ``worse,'' or ``equal''), achieves $33\%$ accuracy.

Figure~\ref{fig:fidelity} provides a complementary view of cost model fidelity. Each point
represents one algorithm applied to one benchmark model. The $x$-coordinate shows the predicted
layout cost normalized to $[0,1]$ across all algorithms for that model (0 = lowest cost, 1 =
highest), and the $y$-coordinate shows the hardware execution time normalized in the same way. A
perfectly accurate cost model would place all points on the diagonal, meaning the algorithm with the
lowest predicted cost also achieves the lowest execution time. To visualize the distribution, we
aggregate the points into a $5 \times 5$ frequency histogram. The off-diagonal mass indicates
disagreement between predicted costs and hardware performance.

This inaccuracy primarily originates from how the cost model handles DMA transfer costs between
device memory and the on-chip scratchpad (\textbf{RQ1}). While the cost
model captures operator execution costs reasonably well,
its transpose cost estimates fail to account for these downstream memory access patterns.

\subsection{Analysis: Why Does a Provably Optimal Algorithm Underperform?}
\label{sect:analysis}

We now discuss why \alg{Rule-Based} can beat an algorithm with an optimality guarantee, and derive a
criterion for when exact solving pays off.

\paragraph{What the solver measures.}
For any algorithm, the hardware execution time is the objective value of its solution plus the error
of the cost model on that solution. \alg{Solver-Based} minimizes the first term exactly, so whenever
it loses on hardware, the loss comes entirely from the second term. A losing solver run therefore
isolates the error of the cost model. A heuristic's losses mix both terms and cannot be read this
way.

\paragraph{Why the error grows under optimization.}
Where the cost model is accurate, the guarantee carries over to hardware: \alg{Solver-Based}
matches or beats every other algorithm on $8$ of the $12$ dense transformers in
Table~\ref{tab:results}. Where it is not, small model errors turn into large regressions, for two
reasons. First, the solver searches the whole
feasible space and selects the configuration with the lowest predicted cost, which favors
configurations whose cost is underestimated. It therefore does not see the average error of the model
but its largest error. Heuristics search a much smaller set of conventional layouts and are exposed
only to the average error. This also explains why \alg{Greedy}, whose objective value is worse than
the solver's, can be faster on hardware (MusicGen: $27.0$ versus $38.5$ million cycles): it starts
from a topological order and changes one operator at a time, so it stays close to conventional
layouts. Second, layout selection runs before tiling, scheduling, and memory allocation, so an
underestimated layout also changes how the tiler splits operators, and the objective does not account
for the resulting DMA transfers (\textbf{RQ1}).

\paragraph{A criterion for exact solving.}
All evaluated algorithms are within a few percent of the objective optimum, while their hardware
execution times differ by up to $4.6\times$ (\textbf{RQ3}). Optimizing the objective more precisely
therefore pays off only while the optimality gap of the heuristic is larger than the error of the
cost model on the workloads of interest. For dense transformers this holds and the solver is the
better strategy; for vision and multimodal models it does not, and no algorithmic improvement is
visible behind the model error. The limit is not specific to exact solvers: any strategy driven by
the formal objective inherits the same cost model, so on the vision and multimodal workloads a
better search algorithm would not help either.

\subsection{The Role of Cost Models}
\label{sect:cost}

Analytical cost models are standard practice in production compilers: XLA and TVM use them for
fusion and scheduling decisions~\cite{xla,tvm18}, Cadence's Xtensa Neural Network Compiler drives its
operator fusion with an analytical profitability model~\cite{xnnc}, and the production Trainium
compiler in our study is built around a cycle-based analytical model. At the same time, accuracy
varies between models: a comparison of an academic analytical model against the industrial XNNC model
reports substantial differences on the same kernels~\cite{droplet}, which matches our observation
that the same model can be reliable for one class of decisions and unreliable for another. Our
evaluation uses the most accurate model available for this hardware target, the production compiler's
own cycle-based model, so the gaps observed in \textbf{RQ4} are not an artifact of a simplified
academic model. \textbf{RQ3} and \textbf{RQ4} thus leave open how a cost model can be made accurate
enough to optimize against. Two established options are to replace the analytical model with a
\emph{learned predictor}, or to sidestep explicit models via \emph{autotuning} or
\emph{profile-guided optimization}.

Learned cost models have shown promising results in TPU performance prediction~\cite{tpugraphs},
throughput estimation~\cite{MendisRAC19}, and Halide scheduling~\cite{halide}, continuing a long
line of work that uses program features and machine learning to steer compiler
decisions~\cite{cummins-stencil,CumminsP0L17,LeatherC20}. In principle, they can capture complex
effects, such as the impact of layout on memory behavior, without explicit modeling. In practice,
they require large amounts of hardware-specific training data, which must be regenerated as
compilers and architectures evolve~\cite{AshouriKCPS19}. They are also hard to debug, since fixing a
wrong prediction often means retraining without clear insight into the cause~\cite{LeatherC20}, and they
transfer poorly across platforms~\cite{CumminsP0L17,autotvm}.

Autotuning and profile-guided approaches measure actual
performance instead~\cite{tvm18,ansor,AshouriKCPS19}. This works when decisions can be evaluated in
isolation, but layout selection does not fit that model, because its decisions propagate through
later passes. Attributing a performance change to an individual layout choice is therefore
unreliable, and each measurement requires compiling and profiling the full model, an expense that
profile-reuse techniques can amortize but not eliminate~\cite{AyupovPP24}.

A third option is to refine the analytical model in place. Small
amounts of profiling data can calibrate individual components of a cost model, such as transpose and
DMA costs, while the overall structure stays intact~\cite{DJB09,NewellP20}. Section~\ref{sect:analysis}
shows where such calibration should aim, since the solver identifies exactly those configurations on
which predicted and measured rankings disagree, and the dominant error in our setting, the DMA cost
of transposes after tiling, is an isolated term that a small profile set can fit. Calibrating a few
analytical terms this way is a bounded and repeatable task, in contrast to the years of manual tuning
behind \alg{Rule-Based}. It also pays off more broadly: every strategy built on the cost model
benefits from a more accurate one, while hand-written rules must be re-tuned for each new model
family and hardware generation.

\subsection{Practical Suggestions for Compiler Developers}
\label{sect:dev}

Our results have several practical consequences.

Layout selection should be written down as an optimization problem even when it is solved
heuristically. A formal objective gives a reference point for debugging and for comparing strategies,
and it makes external solvers usable; without one, layout logic tends to grow into an opaque
collection of rules. But the objective is only as good as the cost model instantiating it. Since even
a provably optimal solution can lose on hardware (\textbf{RQ3}), once the heuristic's optimality gap
is smaller than the model's error, further work is better spent on cost-model fidelity than on the
search algorithm (Section~\ref{sect:analysis}).

How much machinery a layout pass needs depends on the workload. For transformer architectures with
regular data flow, layout preferences are largely uniform and local decisions rarely cascade, so even
\alg{Local} performs well; for vision models with complex spatial data flow the gap between
\alg{Local} and \alg{Rule-Based} is large (Table~\ref{tab:results}). A simple greedy pass may
therefore suffice for transformer-dominated workloads, while vision and multimodal pipelines warrant
domain-specific heuristics.

Finally, an exact solver is useful in the compiler's testing infrastructure even where it is
not the production strategy. It quantifies the optimality gap of the shipping heuristic, supplies
reference layouts for tuning new heuristics, and catches regressions when heuristic cost is compared
against the optimum on small benchmarks. On workloads where the cost model holds up, such as
transformer models, it can serve as the primary strategy as well, though its optimality is always
with respect to the cost model and should be checked against measured execution time first.

\section{Threats to Validity}
\label{sec:threats}

Even though our evaluation used a specific end-to-end compilation stack, the layout constraints we
optimize are shared among modern architectures. Several boundaries nonetheless limit how far the
results carry.

\paragraph{Internal validity.}
Our framework requires a cost model that maps layout configurations to execution costs. In practice,
such models are always approximate. At the point where layout selection occurs in the compilation
pipeline, downstream decisions such as memory allocation and instruction scheduling have not yet
been made, so their effects have to be estimated. As seen in \textbf{RQ4}, provably optimal
solutions under an imperfect cost model may not translate to the best wall-clock performance.
Closing this gap requires sustained investment in cost model accuracy, as we discuss in
Section~\ref{sect:cost}.

\paragraph{External validity.}
Our evaluation is conducted on a single hardware target (AWS Trainium) using one production compiler
(Neuron Graph Compiler). The specific cost tradeoffs, e.g., which operator configurations are cheap
and how expensive transposes are, may differ on other accelerators (e.g., GPUs or TPUs, mobile
NPUs). In particular, as discussed in Section~\ref{sec:conversions}, on CPUs and GPUs a layout
mismatch can often be absorbed by the consuming kernel through loop transformations such as
interchange or through transposed-operand interfaces, making explicit conversions cheaper or
unnecessary; on partitioned systolic array targets this escape hatch does not exist. While our
formulation is hardware-agnostic, the relative performance of the studied algorithms may not fully
generalize to other targets. Similarly, our benchmark suite, though diverse, consists of open-source
models from HuggingFace; proprietary or domain-specific models may exhibit different dataflow graph
structures and layout characteristics.

\paragraph{Construct validity.}
Our formulation captures layout selection in isolation, but production compilers often make layout,
tiling, and operator fusion decisions jointly. Interactions among these passes can shift the cost
landscape: fusing two operators may eliminate an intermediate tensor and its associated transpose,
changing the optimal layout assignment. In the compiler we study, fusion and tiling run \emph{after}
layout selection, so all strategies are evaluated against the same downstream pipeline and the cost
model targets the pre-fusion graph (Section~\ref{sec:experimental_setup}). In pipelines that fuse
operators before assigning layouts, the problem instance would differ.
Additionally, our objective treats operator costs and transpose costs as independent and
additive under sequential execution. Modern accelerators often overlap computation with data
movement or execute independent operators in parallel; incorporating such parallelism would require
moving from a sum-of-costs objective to a scheduling-aware formulation.

\section{Related Work}
\label{sec:related}

\paragraph{Layout selection in ML compilers.}
Modern ML compilers treat layout optimization as important, but approach it heuristically.
XLA~\cite{xla} includes a layout assignment pass
that selects physical layouts (minor-to-major dimension orderings) for tensors in the computational
graph. The algorithm operates in three phases: first assigning locally-optimal layouts to
``influential'' operators such as convolutions and matrix multiplications, then propagating layouts
forward and backward through the graph, and finally inserting copy operations where layout
mismatches remain. This greedy approach provides no formal guarantees and can produce suboptimal
solutions when local decisions conflict. Similarly, TVM~\cite{tvm18} specifies preferred layouts per
operator and inserts layout transformations when producer and consumer expectations diverge. The
oneDNN Graph Compiler~\cite{onednn} handles layouts as part of a hybrid approach combining
expert-tuned kernels with graph-wide optimizations, though layout decisions are primarily driven by
predefined templates rather than systematic optimization.

Several recent systems have developed specialized heuristics for particular deployment scenarios.
SmartMem~\cite{smartmem} addresses layout transformation elimination for mobile GPUs by classifying
operators by layout sensitivity and applying rule-based propagation to reduce conversions. In our
framework, SmartMem can be understood as a greedy heuristic that restricts layout domains, fixes
layouts at sensitive operators by minimizing local operator cost, and propagates layouts through
flexible operators to minimize local transpose cost. NeoCPU~\cite{neocpu} applies similar ideas to
CNN inference on CPUs, selecting efficient blocked formats for convolutions and propagating these
choices through the graph, with layout-oblivious operators absorbing the propagated layout without
conversion cost. ALT~\cite{alt} extends this line of work by jointly optimizing graph-level layouts
and operator-level loop transformations, using reinforcement learning to search the combined space.
When layout requirements conflict between operators, ALT propagates layouts trying to avoid
inserting transposes.

These systems all show that layout matters in practice, but each relies on hand-tuned heuristics, and none defines the problem those heuristics are approximating.

\paragraph{Formal optimization approaches for layout selection.}
Despite the ubiquity of layout decisions in ML compilers, there has been little formal analysis of
their computational complexity. Prior work often notes that layout optimization is ``\NP-hard'' or
``combinatorial'' without proof or qualification~\cite{smartmem,neocpu,alt}. To our knowledge, this
is the first work to formally prove \NP-hardness for layout selection under realistic restrictions
and to identify graph-structural assumptions under which optimal solutions with guarantees are
possible. Our treewidth-based algorithm builds on standard machinery from parameterized complexity.
Courcelle's theorem shows that a broad class of graph problems can be solved in linear time on
graphs of bounded treewidth~\cite{Courcelle90}, and later work extends this to optimization
objectives like ours~\cite{ArnborgLS91}. Treewidth is also familiar in compilers, where the small
treewidth of structured programs enables efficient register allocation~\cite{Thorup98}. These
results show that a polynomial-time algorithm exists, but the generic constructions have very large
constants. We therefore give an explicit dynamic program for the layout selection objective
(Section~\ref{sec:treewidth}), with a concrete state space and runtime bound.

Several works formulate layout-related decisions as constraint problems. VTensor~\cite{vtensor}
introduces a programming framework based on ``virtual tensors'' that abstracts away physical layouts
from the programmer, applying global optimization via integer linear programming (ILP) to minimize
layout transformations. While the authors claim the search space is large and related to \NP-hard
problems, they provide minimal algorithmic details and no formal complexity analysis. Similarly,
Mirage~\cite{mirage} searches layout strategies while generating candidate programs and checks their
equivalence using probabilistic finite-field testing backed by an SMT solver. Rieber et
al.~\cite{oppermann2021joint} address joint program and layout
transformations for specialized hardware using constraint programming, with a focus on enabling
specific operators on hardware with restrictive requirements. Our \NP-hardness results explain why these
systems resort to ILP or \prob{MaxSAT} solvers, and our tractable cases show
when simpler algorithms might suffice.

Complementary work has developed mathematical frameworks for representing layouts. Linear
Layouts~\cite{linearlayouts} models tensor layouts as linear maps over $\mathbb{F}_2$, enabling
generic layout-to-layout conversions, while recent work~\cite{layoutabstractions} unifies CuTe and
Triton layout abstractions using integer set relations. These contributions are largely orthogonal
to ours: they provide expressive representations for layouts or motivate the importance of
layout optimization, while we study the computational problem of selecting among layouts.

\paragraph{Related compiler optimizations.}
Layout selection intersects with other compiler optimizations, such as loop transformations or
scheduling. On general-purpose hardware, the two are partially interchangeable: loop interchange can
recover the locality benefits of a transposed operand without materializing it, while physically
transposing or packing data can expose contiguous accesses that enable auto-vectorization of
transformer kernels~\cite{vectorizing-transformers}; we discuss this interplay, and why it largely
disappears on partitioned systolic arrays, in Section~\ref{sec:conversions}. ALT jointly optimizes
data layouts and loop schedules over a set of layout transformation primitives,
propagating layouts rather than materializing explicit transposes when possible~\cite{alt}.
Autoschedulers such as Halide's tree-search-based scheduler~\cite{halide} and Ansor~\cite{ansor} explore large spaces of schedules using learned cost models that implicitly capture layout effects, but neither formulates nor analyzes the layout selection subproblem.
A related line of research,
originating in feature-based iterative compilation and autotuning, learns from program
characteristics how to set optimization
parameters~\cite{cummins-stencil,CumminsP0L17,AshouriKCPS19}; our per-family findings (\textbf{RQ1})
suggest layout selection as a natural application for such workload-adaptive policy selection.
TpuGraphs dataset~\cite{tpugraphs} provides benchmarks for learning cost models over tensor
computational graphs including layout configurations, documenting the practical importance of layout
decisions and the challenges of accurate cost modeling for tensor programs.

Work on specialized accelerators faces layout optimization challenges due to hardware constraints. AutoSA~\cite{autosa} automatically generates systolic arrays for FPGAs using polyhedral analysis, while FlexSA~\cite{flexsa} and ArrayFlex~\cite{arrayflex} explore flexible systolic array architectures. Spatial~\cite{spatial} provides a language and compiler for application accelerators with explicit constructs for memory hierarchies. Our work is complementary,
since our formulation can serve as a subproblem within such compilers.

\section{Conclusions}
\label{sec:conclusions}

This paper establishes \prob{Layout Selection} as a first-class optimization problem for ML
compilers. We formalized the problem as cost minimization over dataflow graphs, proved its
\NP-hardness, and gave an explicit dynamic program for the tractable bounded-treewidth regime,
alongside a \prob{MaxSAT} encoding for general instances. Our evaluation within a production
compiler for AWS Trainium demonstrates that the formalization is practical. Solver-based methods are
viable for transformer workloads, where they match and sometimes outperform the production
heuristic. On vision and multimodal models they do not. Since the solver is optimal under the
objective, the remaining gap lies in the cost model rather than in the search.

\paragraph{Future directions.}
On the modeling side, extending the formulation to capture joint layout and tiling decisions, or to
account for operator parallelism, would widen its applicability. On the algorithmic side, our
treewidth results suggest hybrid solvers that retain domain-specific heuristics while providing
guarantees on solution quality. On the empirical side, learned cost models trained on hardware
profiling data could improve the fidelity of the objective and close the gap between solver-optimal
and hardware-optimal layouts.

\bibliographystyle{ACM-Reference-Format}
\bibliography{layout.bib}

\newpage
\appendix
\section{Details on Trainium Graph Compiler}
\label{app:compiler}

\begin{figure}[!ht]
    \centering
    \includegraphics[height=5.2cm]{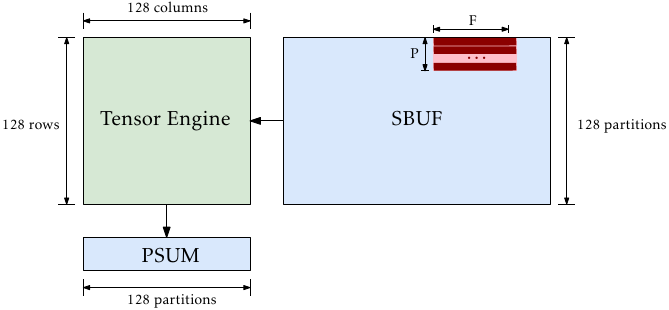}
    \caption{The NeuronCore tensor engine: a $128 \times 128$ systolic array connected to two
        on-chip SRAMs. The state buffer (SBUF) stores input tensors partitioned across 128
        partitions, each
        feeding one row of the array. The partial sum buffer (PSUM) accumulates output results.
        A tensor (red) has layout \layout{P,F}.
        Illustration adapted from~\cite{neuron_compiler}.}
    \Description[Block diagram of the tensor engine with the SBUF and PSUM on-chip
    SRAMs]{A schematic with three rectangular blocks. The tensor engine, shaded green, is a
        square block annotated $128$~rows along its left edge and $128$~columns along its top
        edge. To its right sits the state buffer SBUF, shaded blue and annotated $128$~partitions
        along its right edge; an arrow runs leftward from SBUF into the tensor engine, one
        partition feeding each row of the array. Below the tensor engine sits the partial sum
        buffer PSUM, also shaded blue and annotated $128$~partitions along its bottom edge, with an
        arrow running downward from the tensor engine into it. A tensor resident in SBUF is
        highlighted in red to show its layout \layout{P,F}, with its first dimension spread across
        the partitions and its second dimension laid out along the free axis within each
        partition.}
    \label{fig:trn}
\end{figure}

AWS Trainium~\cite{aws_trainium} is a machine learning accelerator designed to efficiently execute
deep learning workloads using specialized compute engines. Each Trainium device contains multiple
NeuronCores, which perform tensor operations such as matrix multiplication and vector computations.
Each NeuronCore includes a $128 \times 128$ systolic array optimized for matrix multiplication and
convolutions, along with vector and scalar engines for reductions, activations, and other
operations. A key component of the NeuronCore architecture is the \emph{state buffer} (SBUF), a
24\,MB high-bandwidth on-chip scratchpad memory used to store tensors and intermediate results
during execution. Unlike conventional hardware caches, SBUF is software-managed, meaning the
compiler is responsible for organizing and placing data so that the compute engines can access it
efficiently. SBUF is organized into $128$ \emph{independent partitions}, each feeding one row of the
systolic array; refer to Figure~\ref{fig:trn}.

The Neuron Graph Compiler~\cite{neuron_compiler} compiles computation graphs from machine learning
frameworks into executable binaries for NeuronCores. The compiler performs graph-level and
loop-level optimizations to map tensor operations onto hardware instructions, determining how
tensors are tiled, scheduled, and placed in on-chip memory. Because data movement between device
memory and SBUF is expensive, the compiler structures execution so that active tensors remain in
SBUF while operators consume them.

Tensor layout plays a central role because SBUF's partitioned organization exposes hardware
parallelism. Tensors in SBUF are described using two logical axes: a \emph{partition axis}~(P),
which distributes elements across the 128 hardware partitions for parallel execution, and a
\emph{free axis}~(F), which stores the remaining elements within each partition. The dimension
mapped to the partition axis determines the degree of parallelism available to the hardware, while
the free axis is processed sequentially within each partition. Since each partition feeds one row of
the systolic array, the dimension streamed into the array must lie within partitions, not across
them. For instance, computing $C = A \times B$ requires $A$ in \layout{F,P} layout (columns
partitioned across the array rows) and $B$ in \layout{P,F} layout (rows partitioned); the result $C$
is produced in \layout{F,P} layout. Hardware instructions expect specific mappings of tensor
dimensions to these axes, and mismatches between adjacent operators require explicit layout
conversions. To support high hardware utilization, the Neuron compiler must select tensor layouts
that both satisfy architectural constraints and minimize costly layout conversions, making layout
selection an important component of the compilation pipeline.

\section{Proof of Theorem~\ref{thm:treewidth}}
\label{app:treewidth-proof}

We give the formal dynamic program sketched in Section~\ref{sec:treewidth}, followed by its
correctness and runtime analysis.

\begin{proof}
    Let $G=(V,E)$ be the directed dataflow graph and let $H$ be its underlying undirected graph. Let
    $(T,\{X_i\})$ be the given tree decomposition of $H$ with width~$\tw$; transform it into a nice
    decomposition of the same width with root~$r$ in linear time~\cite{Kloks94}. We use dynamic
    programming over the decomposition.

    \paragraph{DP table.}
    For each node $i$ of~$T$, let $V_i$ be the set of vertices occurring in some bag in the subtree
    rooted at~$i$, and let $\Out(X_i)$ be the set of tensors produced by operators currently
    in $X_i$. For each tensor $t \in \Out(X_i)$, let
    $\mathbf{S}=(S_t)_{t\in\Out(X_i)}$ track the set of distinct consumer layouts
    $S_t \subseteq L(t)$ that have been requested by consumer operators strictly inside
    $V_i \setminus X_i$ (i.e., those already forgotten). For a labeling $f$ of the bag $X_i$ with
    $f(o)\in\St(o)$, define $\OPT(i,f,\mathbf{S})$ as the minimum cost for the subtree at $i$. This
    cost accounts for all operator costs belonging to $V_i \setminus X_i$ and all transpose costs
    for tensors produced in $V_i \setminus X_i$. Note that when a tensor's producer is forgotten
    (moves from $X_i$ to $V_i \setminus X_i$), its aggregated consumer set is complete, and its full
    transpose penalty is paid. At the root~$r$, $X_r=\emptyset$ and $V_r=V$, thus
    $\OPT(r,\emptyset,\emptyset)$ equals the global optimum.

    \paragraph{Inductive computation.}
    We compute $\OPT$ bottom-up by node type.

    \begin{itemize}
        \item \textbf{Leaf:} $\OPT(i,\emptyset,\emptyset)=0$.

        \item \textbf{Introduce $v$:} ($X_i=X_j\cup\{v\}$ for some $v\in V\setminus X_j$):
        \[
        \OPT(i,f,\mathbf{S})
        =
        \begin{cases}
            \OPT\bigl(j,\,f|_{X_j},\,\mathbf{S}|_{X_j}\bigr) & \text{if } S_t = \emptyset \text{
                for } t \in \Out(v), \\
            \infty & \text{otherwise.}
        \end{cases}
        \]
        The configuration $f(v)$ is recorded. If $v$ produces a tensor $t$, its tracked consumer set
        $S_t$ must be strictly empty, since no operators in the subtree $V_i \setminus X_i$ could
        have requested it before $v$ was introduced.

        \item \textbf{Forget $v$:} ($X_j=X_i\cup\{v\}$ for some $v\in X_j$):
        Let $f_v = f\cup\{v\mapsto c_v\}$. When $v$ is forgotten, two things happen. First, $v$'s
        input layout requirements are recorded in $\mathbf{S}$. Since $\mathbf{S}$ tracks requests
        from \emph{forgotten} operators, and $v$ becomes forgotten here, the parent state must
        include $v$'s request for each of its input tensors that is still tracked. Thus, for every
        $u \in \In(v) \cap \Out(X_i)$, the parent state $\mathbf{S}$ must satisfy
        $S_u = S'_u \cup \{c_v|_u\}$, and $S_u = S'_u$ for all other $u \in \Out(X_i)$.
        Second, for the tensor $t \in \Out(v)$ produced by $v$, its consumer set is now fully
        determined: it is the union of requests from previously forgotten descendants ($S'_t$) and
        requests from any still-active consumers currently in $X_i$. Let
        $S^{\text{final}}_t = S'_t \cup \{ f(o)|_t \mid o \in X_i, t \in \In(o) \}$. Since $t$ is
        removed from $\Out(X_i)$, $S_t$ no longer exists in $\mathbf{S}$ and we must minimize
        over all feasible subsets $S'_t$. We compute:
        \[
        \OPT(i,f,\mathbf{S})
        =
        \min_{c_v\in \St(v)}\left\{ \min_{\text{compatible } \mathbf{S}'}
        \Bigl\{
        \OPT\bigl(j,\,f_v,\,\mathbf{S}'\bigr)
        + \opcost_{v}(c_v)
        + \sum_{t \in \Out(v)} \sum_{\ell' \in S^{\text{final}}_t}
        \tcost_t\bigl(c_v|_t,\,\ell'\bigr)
        \Bigr\}\right\},
        \]
        where the inner minimum is over all states $\mathbf{S}'$ satisfying the subset relations
        defined above.
        We evaluate this transition forward, so each child entry determines and updates one parent
        entry and is processed once.

        \item \textbf{Join} ($X_i=X_{j_1}=X_{j_2}$):
        \[
        \OPT(i,f,\mathbf{S})
        =
        \min_{\mathbf{S} = \mathbf{S}_1 \cup \mathbf{S}_2}
        \Bigl( \OPT(j_1,f,\mathbf{S}_1)+\OPT(j_2,f,\mathbf{S}_2) \Bigr),
        \]
        where $\mathbf{S}_1 \cup \mathbf{S}_2$ takes the point-wise union $S^{(1)}_t \cup S^{(2)}_t$
        for each tensor $t$, ensuring that identical consumer layouts requested in both subtrees
        merge without double-counting.
    \end{itemize}

    \paragraph{Correctness.}
    We argue by induction on $T$, from the leaves upwards, that $\OPT(i,f,\mathbf{S})$ equals the
    minimum of
    \[
    \sum_{o \in V_i \setminus X_i} \opcost_o\bigl(g(o)\bigr)
    \;+\;
    \sum_{t \in \Out(V_i \setminus X_i)} \;\; \sum_{\ell \in \cns(t)}
    \tcost_t\bigl(\prd(t),\,\ell\bigr)
    \]
    over all labelings $g$ of $V_i$ that agree with $f$ on $X_i$ and induce exactly the request sets
    $\mathbf{S}$. Call this the \emph{invariant}. We use two properties of a nice decomposition. Every
    vertex is introduced exactly once and forgotten exactly once, so the bags containing a vertex $v$
    form a contiguous stretch of $T$ reaching from its introduce node up to the child of its forget
    node. And at a join node, $V_{j_1}\cap V_{j_2}=X_i$, with no edge of $H$ between
    $V_{j_1}\setminus X_i$ and $V_{j_2}\setminus X_i$.

    \emph{Leaf.} $V_i=X_i=\emptyset$, both sums are empty, and the invariant holds with value $0$.

    \emph{Introduce $v$.} Since $v$ is introduced here, $v\notin V_j$, and
    $V_i\setminus X_i = V_j\setminus X_j$. Neither sum changes: no operator leaves the frontier, and
    no tensor's producer does either. The recurrence accordingly copies the child value, so it
    remains to justify the guard. Let $t\in\Out(v)$ and suppose some consumer $o$ of $t$ lay in
    $V_i\setminus X_i$. Then $o$ was forgotten strictly below $i$, so no bag at or above $i$
    contains $o$; but $o$ is adjacent to $v$ in $H$ and must share a bag with it, and every bag
    containing $v$ is at or above $i$, a contradiction. Hence $S_t=\emptyset$ is the only request
    set consistent with the invariant, and entries with $S_t\neq\emptyset$ are correctly assigned
    cost $\infty$.

    \emph{Forget $v$.} Here $V_i=V_j$ and $V_i\setminus X_i=(V_j\setminus X_j)\cup\{v\}$, so relative
    to the child the first sum grows by exactly $\opcost_v(c_v)$ and the second grows by the
    transpose term of the tensor $t\in\Out(v)$, whose producer has just left the frontier. The claim
    to verify is that $S^{\text{final}}_t=\cns(t)$. Every consumer $o$ of $t$ is adjacent to $v$, so
    $o$ and $v$ share a bag; that bag lies in the subtree rooted at $j$, because all bags containing
    $v$ do, and therefore $o\in V_j=V_i$. Each such $o$ is thus either in $X_i$, where its request
    $f(o)|_t$ is read off the bag labeling, or in $V_i\setminus X_i$, where it was forgotten at some
    node below $i$ and the same rule added its request to $S'_t$. So $S^{\text{final}}_t$ lists the
    layouts of $\cns(t)$, each exactly once, and
    $\sum_{\ell'\in S^{\text{final}}_t}\tcost_t(c_v|_t,\ell')$ is precisely the term the objective
    charges for $t$. Minimizing over $c_v$ and over the compatible $\mathbf{S}'$ ranges over exactly
    those labelings of $V_j$ that restrict to $f$ on $X_i$ and induce $\mathbf{S}$.

    \emph{Join.} The two subtrees share no operator outside $X_i$ and no edge between their private
    parts, so every labeling of $V_i$ agreeing with $f$ splits uniquely into labelings of $V_{j_1}$
    and $V_{j_2}$, and the two child costs are sums over the disjoint sets $V_{j_1}\setminus X_i$ and
    $V_{j_2}\setminus X_i$. A layout requested in both subtrees belongs to a tensor still produced in
    $X_i$, whose transpose cost has not been charged yet. Taking the point-wise union instead of a
    multiset sum therefore avoids charging that layout twice once the producer is forgotten.
    Minimizing over all splits $\mathbf{S}=\mathbf{S}_1\cup\mathbf{S}_2$ therefore establishes the
    invariant at $i$.

    At the root, $X_r=\emptyset$ and $V_r=V$, so the invariant makes
    $\OPT(r,\emptyset,\emptyset)$ the optimum of \prob{Layout Selection}.

    \paragraph{Runtime.}
    Each bag has size at most $\tw+1$. The base state $f$ has at most $M^{\tw+1}$ entries per node.
    The augmented state $\mathbf{S}$ tracks a subset $S_t \subseteq L(t)$ \emph{only} for the
    tensors actively produced by operators in $X_i$. This adds at most $(2^{M})^{\tw+1}$
    combinations. A forget node processes at most $M^{\tw+2}2^{M(\tw+2)}$ child entries in
    $O(\tw)$ time each. At a join node, each layout may occur in neither, either, or both child
    request sets, giving at most $4^{M(\tw+1)}$ pairs per bag labeling. Since a nice decomposition
    has $O(\tw|V|)$ nodes, the total running time is
    $O\!\left(|V|\cdot \tw^2 \cdot M^{\tw+2}\cdot 4^{M(\tw+1)}\right)$, which is polynomial in
    $|V|$ for fixed $\tw$ and $M$.

    The optimal assignment itself can be recovered via standard backtracking: at each node, record
    the minimizing configuration, then trace from root to leaves.
\end{proof}

\end{document}